\documentclass[aps,pra,twocolumn,longbibliography,superscriptaddress,10pt]{revtex4-1}

\usepackage{amsmath,amssymb,graphicx,amsmath,amssymb,amsthm,physics,bm,bbm,color,mathtools,anyfontsize,latexsym,epstopdf,scalerel}
\usepackage{appendix}
\usepackage[colorlinks=true,linkcolor=blue,citecolor=green,plainpages=false,pdfpagelabels]{hyperref}
\usepackage[normalem]{ulem}
\usepackage{soul}
\usepackage{newlfont}
\usepackage{amsfonts}
\usepackage{epsfig}
\usepackage[thinc]{esdiff}
\usepackage{times}

\newtheorem{proposition}{Proposition}

\def\tr{\operatorname{tr}}
\newtheorem{theorem}{Theorem}

\newenvironment{prof}{\noindent \textit{Proof sketch:}\ignorespaces}{\hspace*{\fill}$\Box$\medskip}
\DeclareOldFontCommand{\rm}{\normalfont\rmfamily}{\mathrm}
\DeclareMathOperator{\sgn}{sgn}

\DeclareOldFontCommand{\rm}{\normalfont\rmfamily}{\mathrm}
\newcommand{\brad}[1]{( #1|}
\newcommand{\ked}[1]{|#1)}
\newcommand{\bradked}[2]{( #1 | #2 )}
\newcommand{\kedbrad}[2]{|#1 )( #2|}

\newtheorem*{corollary*}{Corollary}

\begin{document}

\title{Family of Exact and Inexact Quantum Speed Limits for Completely Positive and Trace-Preserving Dynamics}

\author{Abhay Srivastav}
\affiliation{Harish-Chandra Research Institute,  A CI of Homi Bhabha National Institute, Chhatnag Road, Jhunsi, Prayagraj  211019, Uttar Pradesh, India}
\affiliation{Optics and Quantum Information Group, The Institute of Mathematical Sciences, A CI of Homi Bhabha National Institute,\\
C.I.T. Campus, Taramani, Chennai 600113, India}

\author{Vivek Pandey}
\affiliation{Institute of Physics, Faculty of Physics, Astronomy and Informatics, Nicolaus Copernicus University, Grudziadzka 5/7, 87-100 Toru\'n, Poland
}

\author{Brij Mohan}
\email{brijhcu@gmail.com}
\affiliation{Nano and Molecular Systems Research Unit, University of Oulu, FI-90014 Oulu, Finland}
\affiliation{Department of Physical Sciences, Indian Institute of Science Education and Research (IISER), Mohali, Punjab 140306, India}

\author{Arun Kumar Pati}
\affiliation{Centre for Quantum Technology,  KIIT, Bhubaneswar, 751024,  India}

\begin{abstract}
Traditional quantum speed limits formulated in density matrix space are generally unattainable for a wide class of dynamics and it is difficult to characterize the fastest possible dynamics. To address this, we present two distinct quantum speed limits in Liouville space for Completely Positive and Trace-Preserving (CPTP) dynamics. The first bound saturates for time-optimal CPTP dynamics, while the second bound is exact for all states and all CPTP dynamics. Our bounds have a clear physical and geometric interpretation arising from the uncertainty relations for operators acting on Liouville space, and the geometry of quantum evolution in Liouville space. We also obtain the form of the Liouvillian, which generates the time-optimal CPTP dynamics that connect the given initial and target states. To illustrate our findings, we show that the speed of evolution in Liouville space bounds the growth of the spectral form factor and Krylov complexity of states, which are crucial for studying information scrambling and quantum chaos. In another important application, we show that our results can help us understand the counter-intuitive phenomenon of the Mpemba effect in non-equilibrium open quantum dynamics, as the minimal relaxation time scale obtained by speed limits is dictated by the eigenmodes of the Liouvillian.
\end{abstract}

\maketitle

In non-equilibrium quantum dynamics, a fundamental problem is to understand the speed at which quantum systems evolve when they are subjected to external environments and fields. Currently, this has become one of the primary areas of research, considering the importance of understanding the speed of many interesting physical phenomena such as the quantum version of the Mpemba effect~\cite{Nava2019, Carollo2021, Chatterjee2023, Rylands2023, Nava2024, Strachan2024, Moroder2024}, relaxation processes in open quantum systems~\cite{Zhou2023, Wang2023}, the growth of spread and Krylov complexity~\cite{Parker2019, Balasubramanian2022, Alishahiha2023, Caputa2024, Nandy2024}, information scrambling~\cite{Vikram2024}, neutrino oscillation~\cite{Jha24}, etc. This notion also holds significant practical relevance in the rapidly developing field of quantum information technologies, as the speed of the dynamics impacts the performance and computational power of quantum devices ranging from quantum computers~\cite{Lloyd2000} and quantum sensors~\cite{Beau2017, Herb2024} to quantum batteries~\cite{Campaioli2017, Bera2020, Mohan2021, Mohan2022} and quantum thermal machines~\cite{Bera2022}. Hence, it is crucial to understand the general properties of quantum dynamics that affect the evolution speed of the quantum system.

In physical processes, the rate at which the state of a quantum system changes is referred to as the speed of evolution. This speed of evolution constrains the rate of change of distinguishability between initial and time-evolved states, establishing a lower bound on the evolution time known as quantum speed limit (QSL), thus representing inherent dynamical limitation.
The first QSL was the Mandelstam-Tamm (MT) bound \cite{Mandelstam1945}, which was derived for unitary dynamics generated by a time-independent Hamiltonian using the Robertson uncertainty relation~\cite{Robertson1929}, and later using the geometry of quantum evolution \cite{Anandan1990}. The MT bound states that the evolution time is lower bounded by the ratio of the Hilbert space angle between the initial and final states to the variance of the driving Hamiltonian. This bound provides an unambiguous meaning to the time-energy uncertainty relation~\cite{Aharonov1961, Busch2008} and saturates for the optimal unitary dynamics \cite{Carlini2006, Mostafazadeh2009}. Subsequent developments have also generalized the MT bound to time-dependent Hamiltonians~\cite{Pfeifer1993}. Recently, the exact form of the MT bound has been obtained for unitary dynamics~\cite{Pati2023}. The MT bound has also been extended to the unitary dynamics of operators~\cite{Pintos2022, Hamazaki2022, Mohan2022, Carabba2022, Hamazaki2024}, superoperators~\cite{Hornedal2022}, entanglement~\cite{Shrimali2022, Mohan2023, Mohan2024}, as well as various classical dynamics~\cite{Nicholson2020, Das2023}, and non-Hermitian quantum dynamics~\cite{Thakuria2024, Nishiyama2024, Hornedal2024}. Over time, the MT bound has found numerous applications in rapidly growing quantum technologies~\cite{Lloyd2000, Murphy2010, Campaioli2017, Lam2021}, and several experiments have verified it~\cite{Ness2021, Wu2024}. Recently, Ref.~\cite{Johnsson2023} studied the exact and inexact speed limits for change of unitary operators rather than of states or observables.

In the last decade, various generalized QSLs have been formulated for more general types of quantum dynamics using mathematical inequalities~\cite{Campo2013, Deffner2013, Sun2015, Uzdin_2016, Teittinen_2019, Pires2024} and the geometry of density matrix space~\cite{Taddei2013, Pires2016, Campaioli2019, Funo2019, Connor2021, Nakajima2022, Lan2022, Mai2023,  Rosal2024, Mai2024}. For a comprehensive overview of historical and recent advancements in speed limits, see Refs.~\cite{Deffner2017, Gong2022}. The non-geometric QSLs are, in general, loose as they are obtained using a sequence of various inequalities. The geometric QSLs saturate for dynamics that transport the state along the corresponding geodesic. However, in general, finding the general form of the dynamics (i.e., Kruas operators or Liouvillian for given initial and target states) that takes the system along the geodesic for these QSLs is not known. The equation for the geodesic for some of these QSLs is given in the Appendix of Ref.~\cite{Connor2021}. In general, these geometric QSLs are not attainable for a wide class of dynamics (apart from the dynamics that takes the system along the geodesic) due to (i) the presence of classical uncertainty within the states and (ii) some parts of the generator of the dynamics may not contribute to the change of the distinguishability. Moreover, these QSLs do not resemble the MT bound, as they do not depend on the variance of the generator of the dynamics. Consequently, they may not fully capture the dynamical properties needed to control the pace of dynamics, as the MT bound does for unitary dynamics~\cite{Caneva2009, Poggi2013, Deffner2017, Ansel2024}. Most of these QSLs are challenging to compute as the system's dimension or the complexity of the dynamics increases. This difficulty arises because they generally require estimating norms of non-Hermitian operators, the quantum Fisher information, and the square roots of density matrices. Furthermore, these speed limits provide a lower bound on the evolution time rather than an exact relation. Thus, the true generalization of the MT bound, which offers a time-energy-dissipation uncertainty relation in its exact form for Completely Positive Trace-Preserving (CPTP) dynamics, remains an important question.

In this article, we prove exact and inexact quantum speed limits for CPTP dynamics for finite-dimensional quantum systems in Liouville space. First, we derive an inexact quantum speed limit using the uncertainty relation for non-Hermitian operators acting on Liouville space. This provides a lower bound on the evolution time, defined as the ratio of the Liouville space angle between the initial and final states to the time-averaged variance of the Liouvillian. We show that this bound is saturated for time-optimal CPTP dynamics, which drives the system along the geodesic in Liouville space. Further, we derive the exact bound on the evolution time using the exact uncertainty relation for operators acting on Liouville space. We show that the evolution time can be exactly expressed as the ratio of the analogous Wootters distance in Liouville space to the variance of the non-classical part of the Liouvillian. We have also provided geometrical interpretations for both the bounds. Additionally, we show two important applications of our results. First, we show that our speed limit sets bounds on the spectral form factor and Krylov complexity of states, which are important quantities for studying information scrambling and quantum chaos. Second, our speed limits reveal why states far from equilibrium can relax more quickly than those near equilibrium. This insight naturally explains the Mpemba effect in open quantum dynamics. 

\textbf{\it Liouville Space:--} For every linear operator \( A = \sum_{ij} \alpha_{ij}\ketbra{i}{j}\)  acting on a $d$-dimensional Hilbert space  \(\mathsf{H}_d\), there exists a vector \( \ked{A} = \sum_{ij} \alpha_{ij} \ket{j} \otimes \ket{i} \). These vectorized linear operators form a Hilbert space \( \mathsf{H}_d \otimes \mathsf{H}_d \) also known as Liouville space, with the inner product \( \bradked{A}{B} := \tr{(A^\dagger B)} \)~\cite{Sarandy_2004, Gyamfi2020}. We denote the linear operators that act on the vectors in Liouville space by \(\mathcal{A}, \mathcal{B}, \ldots\), etc. In Liouville space, arbitrary CPTP dynamics is described by the following master equation (see Appendix~\ref{appendixA}):
\begin{equation} \label{EOM}
    \frac{d}{dt}\ked{{\tilde\rho_t}} = \left(\mathcal{L}_t - \frac{1}{2} \left\{ (\tilde\rho_t|\mathcal{L}_t|\tilde\rho_t) + (\tilde\rho_t|\mathcal{L}_t^\dagger|\tilde\rho_t) \right\} \right)\ked{{\tilde\rho_t}},
\end{equation}
where \( \mathcal{L}_t \) is the generator of the dynamics and is called Liouvillian, and \( \mathcal{L}^\dagger_t \) is its adjoint. Here \( \ked{\tilde\rho_t} := \frac{\ked{\rho_t}}{\norm{\ked{\rho_t}}} \) represents the normalized state vector corresponding to the time-evolved density matrix \( \rho_t \), with \( \norm{\ked{\rho_t}} = \sqrt{\tr{(\rho_t^2)}} \) representing the Hilbert-Schmidt norm of \( \rho_t \).  For the rest of the article, we drop the subscript '$t$' from \( \mathcal{L}_t \) and assume that the Liouvillian is time-dependent unless stated otherwise. We assume $\hbar = 1$ throughout this work.  In this work, we refer an equality relation as "exact bound" and inequality relation as "inexact bound". Moreover, the inexact bounds that are attainable are referred to as tight bounds. We now present two uncertainty relations for operators acting on Liouville space: 

\begin{proposition}\label{InEUN}
For any two non-Hermitian operators acting on Liouville space, \(\mathcal{A}\) and \(\mathcal{B}\), there exists an inexact uncertainty relation given by 
\begin{equation}\label{uncertreln}
 (\Delta \mathcal{A})^2(\Delta \mathcal{B})^2 \geq \left| \tr(\mathcal{A}^\dagger \mathcal{B}\mathcal{P}) - \tr(\mathcal{A}^\dagger\mathcal{P}) \tr(\mathcal{B}\mathcal{P}) \right|^2,
\end{equation}
where \(\mathcal{P} := \kedbrad{\tilde{\rho}}{\tilde{\rho}}\) represents the projection operator acting on  Liouville space, associated with  the state vector \(\ked{\tilde{\rho}}\), and \(\tr(\mathcal{O}\mathcal{P})\) and \((\Delta \mathcal{O})^2 = \tr(\mathcal{O}^\dagger \mathcal{O}\mathcal{P}) - \tr(\mathcal{O}^\dagger\mathcal{P}) \tr(\mathcal{O}\mathcal{P})\) are the average and the variance of \(\mathcal{O}\) in the state vector \(\ked{\tilde{\rho}}\).
\end{proposition}

\begin{proposition}\label{EUN}
For any Hermitian operator \(\mathcal{A}\) and a non-Hermitian operator \(\mathcal{B}\), acting on Liouville space, there exists an exact uncertainty relation for all state vectors \(\ked{\tilde{\rho}}\) given by 
\begin{equation}\label{exactuncrln}
 \delta_{\mathcal{B}}\mathcal{A} \Delta \mathcal{B}_{nc} = \frac{1}{2},
\end{equation}
where \(\mathcal{B}_{nc} = \mathcal{B} - \mathcal{B}_{cl}\) is the non-classical (non-diagonal) part of \(\mathcal{B}\) in the basis of \(\mathcal{A}\), i.e., \(\{\ked{a_i}\}\), and \(\delta_{\mathcal{B}}\mathcal{A}\) is the uncertainty of the operator \(\mathcal{A}\) in the state vector \(\ked{\tilde{\rho}}\), defined as $(\delta_{\mathcal{B}}\mathcal{A})^{-2} := \sum_i \frac{ \brad{a_i}\mathcal{G}\ked{a_i}^2}{(a_i|\mathcal{P}|a_i)}$, $\mathcal{B}_{cl} := \sum_i \kedbrad{a_i}{a_i} \frac{(a_i|\mathcal{C}|a_i)}{(a_i|\mathcal{P}|a_i)}$, where \(\mathcal{G} := (\mathcal{B}\mathcal{P} + \mathcal{P} \mathcal{B}^\dagger - \mathcal{P}\tr[\mathcal{P}(\mathcal{B} + \mathcal{B}^\dagger)])\), \(\mathcal{C} := \frac{i}{2}(\mathcal{B}\mathcal{P} - \mathcal{P} \mathcal{B}^\dagger)\) and $\mathcal{B}_{cl}$ is the classical part of $\mathcal{B}$ that commutes with $\mathcal{A}$.
\end{proposition}

The proof and details of the above uncertainty relations are provided in Appendix~\ref{appendixA}. Next, we derive the inexact QSL for arbitrary CPTP dynamics in finite-dimensional quantum systems.
\begin{theorem}\label{MT}
For a $d$-dimensional quantum system, the lower bound on the evolution time required to evolve a given state \(\rho_0\) to a final state \(\rho_T\) under a CPTP dynamics generated by the Liouvillian \(\mathcal{L}\) is given by
\begin{equation}\label{LMT}
T \geq \frac{\Theta(\rho_0, \rho_T)}{\langle\!\langle \Delta \mathcal{L} \rangle\!\rangle_{T}},
\end{equation}
where \(\Theta(\rho_0, \rho_T) := \arccos{(\tilde{\rho}_0|\tilde{\rho}_T)}\) is Liouville space angle between the state vectors \(\ked{\tilde{\rho}_0}\) and \(\ked{\tilde{\rho}_T}\), corresponding to the initial and final states \(\rho_0\) and \(\rho_T\), respectively, \(\Delta \mathcal{L}\) is the variance of the Liouvillian in the time-evolved state vector \(\ked{\tilde{\rho}_t}\), and \(\langle\!\langle X_t \rangle\!\rangle_{T} := \frac{1}{T} \int_{0}^{T} \, dt X_t \) is the time average of \(X_t\).
\end{theorem}
\begin{prof}
The upper bound on the rate of change of the distinguishability between the initial state vector \(\ked{\tilde{\rho}_0}\) and the time-evolved state vector \(\ked{\tilde{\rho}_t}\), using Eq.~\eqref{EOM} and the uncertainty relation given in Proposition~\ref{InEUN}, can be written as $\left|\frac{d}{dt}(\tilde{\rho}_0|\tilde{\rho}_t)^2\right| \leq 2\Delta\mathcal{L}\Delta\mathcal{P}_0$, where \(\mathcal{P}_0 = \kedbrad{\tilde{\rho}_0}{\tilde{\rho}_0}\). Integrating this inequality and using \(\left| \int f(t) \, dt \right| \leq \int |f(t)| \, dt\), we obtain the bound in Eq.~\eqref{LMT} (see Appendix~\ref{appendixB}).
Alternatively, Theorem~\ref{MT} can be proven using the geometry of quantum evolution in Liouville space. In Liouville space, the infinitesimal distance \(dS\) along the evolution path traced out by the state vector \(\ked{\tilde{\rho}_t}\) is \(dS^{2} = \Delta \mathcal{L}^{2} dt^{2}\). This equation provides the speed of evolution of the quantum system in Liouville space, i.e., \(dS/dt = \Delta \mathcal{L}\), which says that uncertainty of the Liouvillian drives the system forward in time. Integrating \(dS/dt = \Delta \mathcal{L}\) with respect to time and recognizing the fact that \( S \geq \Theta(\rho_0, \rho_{T}) \), yields the bound in Eq.~\eqref{LMT}. Moreover, the bound in Eq.~\eqref{LMT} can also be expressed in terms of the Kraus operators instead of the Liouvillian, and it can be generalized beyond CPTP dynamics (see Appendix~\ref{appendixC}).
\end{prof}

For Lindblad dynamics, the evolution of state is given as $\dot{\rho_t}=-i[H_t,\rho_t]+ \mathcal{D}(\rho_t)$, where $\mathcal{D}(\rho_t)=\sum_{k=1}^{d^2-1}\gamma^{k}_{t}\left(L^{k}_{t}\rho_t L^{k\dagger}_{t}-\frac{1}{2}\{L^{k\dagger }_{t}L^{k}_{t},\rho_t\}\right)$ is the dissipator. Hence, the square of the speed of evolution can be written as
\begin{eqnarray}\label{BM}
    \Delta\mathcal{L}^2&=&\frac{1}{\tr(\rho_t^2)}\left(\tr(\dot{\rho}_t^\dagger\dot{\rho}_t)-(\tr(\rho_t\dot{\rho}_t))^2\right)\nonumber\\
    &=&\frac{1}{\tr(\rho_t^2)}\left(2\tr(\rho_t^2 H_t^2-\rho_t H_t\rho_t H_t) + \tr(\mathcal{D}(\rho_t)^2)\right .\nonumber\\
 && +2i\tr(\rho_t[H_t,\mathcal{D}(\rho)])- \tr(\rho_t \mathcal{D}(\rho_t))).
\end{eqnarray}

It is worth mentioning that the term ``$2i\tr(\rho_t[H_t,\mathcal{D}(\rho)])$" in Eq.~\eqref{BM} vanishes if the Lindblad operators $L_{k}$ are Hermitian and commute with the Hamiltonian $H_{t}$~\cite{Brody2019}.

The Liouvillian (obtained by vectorizing Lindblad master equation) can be decomposed into two parts as \(\mathcal{L} = -i\mathcal{L}_{H} + \mathcal{L}_{D}\), where $\mathcal{L}_H :=\left(\mathbb{I}\otimes H_t-H_{t}^T\otimes\mathbb{I}\right)$ is associated with the Hamiltonian $H_t$, and $\mathcal{L}_D :=\sum_{k} \gamma^k_{t} \left(L^{k*}_{t}\otimes L^{k}_t -\frac{1}{2}\left(\mathbb{I}\otimes L^{k\dagger}_{t}L^{k}_{t}+(L^{k\dagger }_{t}L^{k}_{t})^T\otimes\mathbb{I}\right)\right)$ is associated with the dissipative part of the Lindblad master equation, where $\{\gamma^k_t\}$ are the decay rates and $\{L^k_t\}$ are the Lindblad operators. Using this decomposition, the square of the speed of evolution can be written as (see Appendix~\ref{appendixB})
\begin{equation}
    \Delta\mathcal{L}^2={\Delta\mathcal{L}_H^2+\Delta\mathcal{L}_D^2+i \tr\left([\![\mathcal{L}_H,\mathcal{L}_D]\!]\mathcal{P}_t\right)},\label{speed}
\end{equation}
where \(\Delta\mathcal{L}_{H}\) and \(\Delta\mathcal{L}_{D}\) are the speeds corresponding to the unitary and the dissipative parts of the dynamics, respectively, and \([\![\mathcal{L}_H,\mathcal{L}_D]\!]:=(\mathcal{L}_H\mathcal{L}_D-\mathcal{L}^\dagger_D\mathcal{L}_H)\)
represents the non-commutativity of the unitary and the dissipative parts of the generator. The last term in Eq.~\eqref{speed} captures the nontrivial correlation between the unitary and dissipative speeds and is upper bounded by \(2\Delta \mathcal{L}_{H} \Delta \mathcal{L}_{D}\). Since \(\Delta\mathcal{L}_{H}\) is related to the uncertainty in energy and \(\Delta\mathcal{L}_{D}\) is related to dissipation, the bound in Eq. \eqref{LMT} can be regarded as time-energy-dissipation uncertainty relation. Therefore, the bound can be seen as the generalization of the original Mandelstam-Tamm bound, which was given for unitary dynamics~\cite{Mandelstam1945, Anandan1990}. Alternatively, the Liouvillian can be decomposed into an anti-Hermitian and a Hermitian part as \(\mathcal{L} = -i\mathcal{L}_{+} + \mathcal{L}_{-}\), where \(\mathcal{L}_{+}=\mathcal{L}_{H} +i\left(\mathcal{L}_{D}-\mathcal{L}^\dagger_{D}\right)/2\) and \(\mathcal{L}_{-}=\left(\mathcal{L}_{D}+\mathcal{L}^\dagger_{D}\right)/2\) are associated with the reversible and irreversible parts of the dynamics, respectively. This distinction is based on the fact that the dynamics associated with $\mathcal{L}_+$ preserves the purity and von Neumann entropy of the state; on the other hand, the non-unitary part of the dynamics associated with $\mathcal{L}_-$ may change the purity and von Neumann entropy of the state. This allows us to express the speed of evolution in terms of the speeds associated with the reversible and irreversible parts of the dynamics, similar to Eq.~\eqref{speed}.

 For Lindblad dynamics, the speed of evolution can also be expressed in terms of the eigenvalues and eigenvectors (decay modes) of the Liouvillian, provided that the Liouvillian is time-independent (see Appendix~\ref{Exp}). This suggests that by eliminating the slowly decaying modes of the Liouvillian, either by selecting an appropriate initial state or employing the strategies discussed in Refs.~\cite{Carollo2021, Kochsiek2022}, the speed of evolution can be enhanced. If only the fastest decaying mode with a real eigenvalue contributes to the dynamics, the system approaches the steady state at the maximum possible speed. Moreover, we can express $T_{QSL}=\Theta(\rho_0, \rho_T)/\langle\!\langle \Delta \mathcal{L} \rangle\!\rangle_{T}$ (lower bound in Eq.~\eqref{LMT}) in terms of the eigenmodes of Liouvillian, which suggests that the minimal time scale of a system is dictated by the eigenmodes of Liouvillian. For details see Appendix~\ref{Exp}.

For a dissipative process, the speed of evolution of a quantum system is time-dependent, even for a time-independent Liouvillian, because the Liouvillian is generally non-Hermitian. Consequently, under such a process, the norm is not preserved, i.e., \(\brad{\rho_{0}}e^{\mathcal{L}^\dagger t}e^{\mathcal{L} t}\ked{\rho_{0}} \neq \bradked{\rho_{0}}{\rho_{0}}\). Hence, under a dissipative process, quantum systems accelerate or decelerate over time. Therefore, to estimate the lower bound in Theorem~\ref{MT}, we typically need to solve for the dynamics.  However, we can use the inequality, \(\Delta \mathcal{L} \leq \|\mathcal{L}\|_{\text{op}}\) (where, $\|\mathcal{L}\|_{\text{op}}$ is the square root of the largest eigenvalue of $\mathcal{L}^{\dagger}\mathcal{L}$) to obtain the following lower bound on evolution time: \(T \geq {\Theta(\rho_0, \rho_T)}/{\langle\!\langle \norm{\mathcal{L}}_{\text{op}} \rangle\!\rangle_{T}}\), where $\langle\!\langle X \rangle\!\rangle_{T}$ is the time-averaged $X$. Estimating the above bound does not require solving for the dynamics. Moreover, for a time-independent Liouvillian, we have $\langle\!\langle \norm{\mathcal{L}}_{\text{op}} \rangle\!\rangle_{T}=\norm{\mathcal{L}}_{\text{op}}$, thus the above bound can be expressed as
\(T \geq \Theta(\rho_0, \rho_T)/\norm{\mathcal{L}}_{\text{op}}\), where $\norm{\mathcal{L}}_{\text{op}}=\sqrt{\lambda_{max}}$ and $\lambda_{max}$ is the largest eigenvalue of $\mathcal{L}^{\dagger}\mathcal{L}$. For CPTP dynamics generated by time-independent Liouvillian, the eigenvalue $\lambda_{max}$ corresponds to the fastest decaying eigenmode because the absolute value of its real part, $\abs{\Re{\lambda_{max}}}$, is the largest of all the eigenvalues of Liouvillian (see Appendix~\ref{Exp} for details). Hence, this bound provides a sensible minimal timescale for an open system to reach a steady state if $\rho_{T}=\rho_{ss}$. Important to note that this bound is computationally costly as it requires to diagonalize the Liouvillian to find its eigenvalues. However, using the inequality ${\norm{\mathcal{L}}_{\text{op}}} \leq {\norm{\mathcal{L}}_{\text{HS}}}={\sqrt{\tr{(\mathcal{L}^\dagger\mathcal{L})}}}$, we obtain the following bound 
\begin{equation}
T \geq \frac{\Theta(\rho_0, \rho_T)}{\sqrt{\tr{(\mathcal{L}^\dagger\mathcal{L})}}},
\end{equation}
which is arguably computationally cheaper than previous bounds as well as do not require to solve the dynamics.

\begin{corollary*}
Any CPTP dynamics generated by a time-independent Liouvillian $\mathcal{L}$ will saturate the bound in Eq.~\eqref{LMT} if the time-evolved state satisfies the following equation:
\begin{equation}\label{Line}
\rho_t = e^{\mathcal{L}t}(\rho_0) = P_t\rho_0 + (1 - P_t)\rho_0^{\perp},
\end{equation}
where $\rho_0$ and $\rho_0^\perp$ denote the initial state and its orthogonal complement (fixed in time), respectively, and $P_t=\frac{\tr(\rho_0\rho_t)}{\tr(\rho_0^2)}$ is the relative purity between the initial and time-evolved states.
\end{corollary*}
The proof is provided in Appendix~\ref{appendixD}. The vectorized form of Eq.~\eqref{Line} represents the geodesic equation in Liouville space. It says that the time-evolved state $\ked{\rho_t}$ traces a line in the two-dimensional subspace spanned by $\{\ked{\rho_0},\ked{\rho^{\perp}_0}\}$. Moreover, the relation $S = \int_{0}^{T} dt \, \Delta \mathcal{L}$ makes it evident that the evolution time \( T \) and the length of the total path \( S \) are monotonic functions of each other. Hence, when the system evolves through the shortest path (geodesic), it naturally takes minimal time~\cite{Mostafazadeh2009}. Therefore, the CPTP dynamics that saturates the bound in Eq.~\eqref{LMT} is also time-optimal. In general, the Liouvillian that generates the time-optimal CPTP dynamics given by~Eq.~\eqref{Line}, can be written as: 
\begin{equation} \label{dynsat}
\mathcal{L}=\gamma\left(U^{*}\otimes U-\mathbb{I}\otimes\mathbb{I}\right),
\end{equation}
 where $\gamma$ is Weiskopf-Wigner decay constant and $U$ is a time-independent unitary such that $\rho_0^{\perp}=U\rho_0U^{\dagger}$ (see Appendix~\ref{appendixD}). Since $U$ is not unique, there can be a family of Liouvillian that generates the same optimal CPTP dynamics for the considered initial and final states. Moreover, we observe that the CPTP dynamics that traces the geodesic cannot be unitary as it changes the purity of state with time. \\

While the bound in Eq.~\eqref{LMT} saturates for the optimal CPTP dynamics that transports the state of a quantum system along the geodesic, there are instances where estimating the exact duration of a physical process is more crucial than knowing just the lower bound. Hence, next we derive the exact quantum speed limit for CPTP dynamics.

\begin{theorem}\label{ExactMT}
For a $d$-dimensional quantum system, the exact time and a lower bound on the exact time required to evolve a given state $\rho_0$ to a final state $\rho_T$ under a CPTP dynamics generated by Liouvillian $\mathcal{L}$ are given by
\begin{equation}\label{EMT}
T= \frac{ l(\rho_0,\rho_{T})}{\langle\!\langle\Delta\mathcal{L}_{\rm nc}\rangle\!\rangle_{T}} \geq \frac{ \Theta(\rho_0,\rho_{T})}{\langle\!\langle\Delta\mathcal{L}_{\rm nc}\rangle\!\rangle_{T}},
\end{equation}
when $\mathcal{L}_{\rm nc}$ is the non-classical part of $\mathcal{L}$ and is computed using the basis 
$\{\ked{a_{i}}\}$ (where the initial state vector $\ked{\tilde{\rho}_0}$ is included in this basis). And using the expansion $\ked{\tilde{\rho}_t} = \sum_{i=0}^{d^2-1} c_{i}\ked{a_{i}}$, we can define the length of the path traced out by the real vector $\ked{\sigma_t}=\sum_{i=0}^{d^{2}-1}\vert c_i\vert\ked{a_i}$ in Liouville space as $l(\rho_0,\rho_T)=\int_{0}^{T}\sqrt{\bradked{\dot{\sigma}_t}{\dot{\sigma}_t}}dt$ where $\ked{\dot{\sigma}_t}= \frac{d}{dt}\ked{\sigma_t}$.
\end{theorem}

Considering $\mathcal{A} = \mathcal{P}_0 = \kedbrad{\tilde{\rho}_0}{\tilde{\rho}_0}$ and $\mathcal{B} = \mathcal{L}$ in the exact uncertainty relation given in Proposition~\ref{EUN} and using Eq.~\eqref{EOM}, we can prove the Theorem~\ref{ExactMT} (see Appendix~\ref{appendixE}). In Theorem~\ref{ExactMT}, $l(\rho_0, \rho_T)$ can be regarded as the analog of the Wootters distance in Liouville space~\cite{Wootters1981} and is lower bounded by the geodesic distance, $l(\rho_0, \rho_T) \geq \Theta(\rho_0, \rho_T)$. This inequality sets a lower bound in Eq.~\eqref{EMT}. Interestingly, when the evolution is confined to a two-dimensional Liouville space and if the relative purity $P_t$ monotonically decreases, then $l(\rho_0, \rho_T) = \Theta(\rho_0, \rho_T)$, leading to the saturation of the inequality in Eq.~\eqref{EMT}. Moreover, $\Delta\mathcal{L}_{\rm nc}$ can be regarded as the refined speed of evolution because it is obtained by removing the classical uncertainty $\Delta\mathcal{L}_{\rm cl}$ from the evolution speed. This is due to the fact that the classical part of $\mathcal{L}$, i.e., $\mathcal{L}_{\rm cl}$ does not contribute to the rate of change of distinguishability between the initial and time-evolved states as it commutes with $\mathcal{P}_0$. Thus we can discard $\mathcal{L}_{\rm cl}$ and only consider the non-classical part of $\mathcal{L}$, i.e., $\mathcal{L}_{\rm nc}$. Note that $\mathcal{L}_{\rm cl}$ and $\mathcal{L}_{\rm nc}$ are computed using Proposition~\ref{EUN}, and $\mathcal{L}_{\rm cl}$ is always anti-Hermition; thus classical uncertainty can only be removed from the variance of the anti-Hermitian part of the Liouvillian. The refined speed of evolution can also be decomposed into unitary and dissipation speeds or reversible and irreversible speeds, similar to Eq.~\eqref{speed}. Thus, bound in Eq.~\eqref{EMT} can be regarded as the exact time-energy-dissipation uncertainty relation or the exact generalization of the MT bound. The lower bound presented in Theorem~\ref{ExactMT} is tighter than the bound of Theorem~\ref{MT} because  $\Delta\mathcal{L}_{\rm nc}^{2} \leq  \Delta\mathcal{L}_{\rm nc}^{2} +  \Delta\mathcal{L}_{\rm cl}^{2}  = \Delta\mathcal{L}^{2}$. Interestingly, the inequality in Eq.~\eqref{EMT} can also be obtained using similar methods as in the proof of Theorem~\ref{MT} (see Appendix~\ref{appendixE}). Note that to obtain tighter QSLs, we normalize the state by its purity and exclude the classical part from the generator of the dynamics.

Next, we present two important applications of our speed limits.

\begin{figure*}[htp]
\centering
\includegraphics[width=5.3cm]{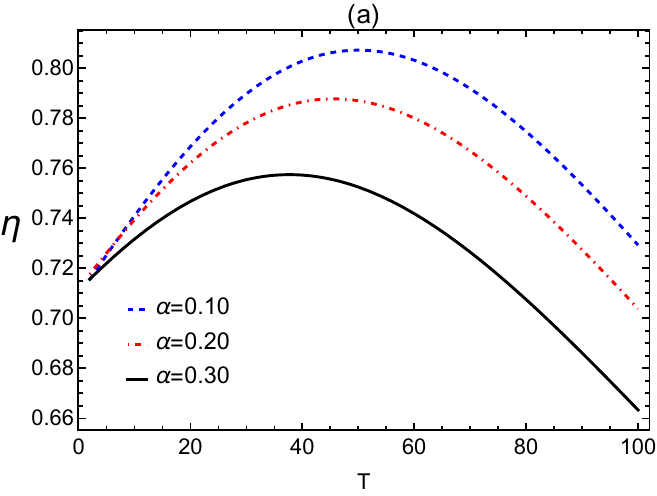}
\space
\space
\space
\includegraphics[width=5.3cm]{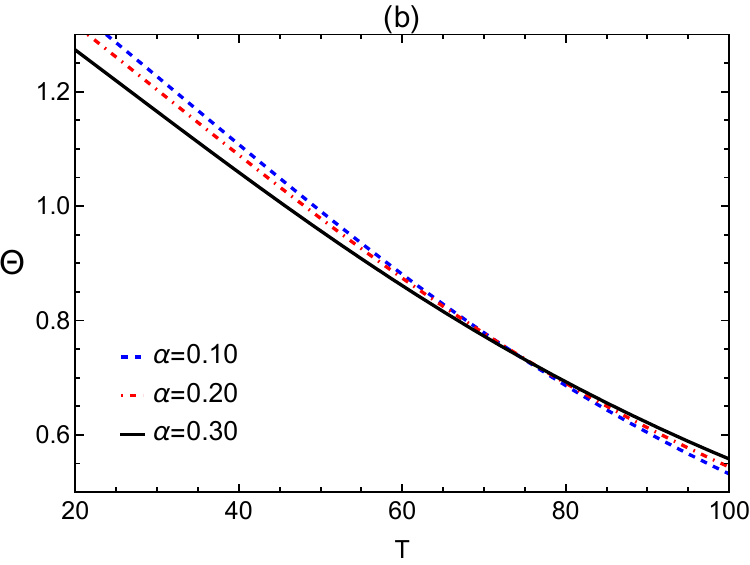}
\space
\space
\space
\includegraphics[width=5.5cm]{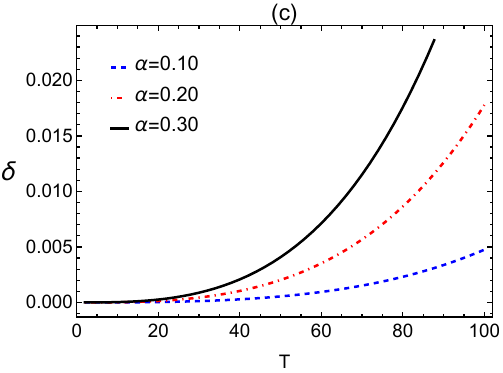}
 \caption{{\bf Mpemba effect in a qubit system under the action of thermal environment.} The computations are performed for various values of $\alpha$ (which corresponds to different initial states) and the parameter $\gamma = 0.01$. The operator-norm of Liouvillian is $\norm{\cal{L}}_{\rm op}= \sqrt{2}\gamma$. (a) Speed efficiency $\eta$ for different initial states. (b) Distance between different non-equilibrium states and the steady state, i.e., $\Theta(\rho_{T},\rho_{ss})$. (c) Tightness of the bound given in Eq.~\eqref{LMT} for different initial states.}
\label{Mpemba}
\end{figure*}
\textbf{\it  Relation between spectral form factor and Krylov complexity of states, and speed of evolution:---} The spectral form factor (SFF) and Krylov complexity of states are two key quantities to study information scrambling and quantum chaos (see Refs.~\cite{Parker2019, Balasubramanian2022, Xu2021, Azcona_2022, Roubeas2023, Alishahiha2023, Caputa2024, Nandy2024}). 
For a system under the action of a CPTP map $\mathcal{E}_t$, the spectral form factor is defined as ${\rm SFF}(t) := \tr \left( \rho_{\beta}\mathcal{E}_{t}(\rho_\beta)\right) $, where $\rho_\beta=\ketbra{\psi_{\beta}}{\psi_{\beta}}$, $\ket{\psi_{\beta}} = Z^{-\frac{1}{2}}_{\beta}\sum_{n}e^{\frac{-\beta E_{n}}{2}}\ket{n}$ is the coherent Gibbs state, $E_n$ is the eigenvalue of the system Hamiltonian, and $Z_{\beta} = \sum_{n} e^{-\beta E_{n}}$. Using the bound in Eq.~\eqref{EMT}, we obtain the following relation
\begin{equation}\label{SFFB}
  \arccos(\frac{{\rm SFF}(T)}{\sqrt{\tr \left[\mathcal{E}_{T}(\rho_\beta)\right]^2}}) \leq \int_{0}^{T}\, dt \Delta \mathcal{L}_{\rm nc}.
\end{equation}
The time evolution of a density operator under unitary dynamics can be written in the Krylov space as $\ked{\rho_t}=\sum_n i^n \phi_n(t) \ked{K_n}$, where $\{\ked{K_n}\}$ form the Krylov basis, $\phi_n(t)$ is the transition amplitude satisfying $\sum_n\abs{\phi_n(t)}^2=1$, and $\phi_n(0)=\delta_{n0}$. The growth of the density operator is quantified by the Krylov complexity~\cite{Parker2019, Balasubramanian2022, Alishahiha2023, Caputa2024, Nandy2024},
$ C_{K}(t) :=\tr(\mathcal{P}_t\mathcal{S})= \sum_n n\abs{\phi_{n}(t)}^2$,
where $\mathcal{P}_t=\kedbrad{\rho_t}{\rho_t}/\tr(\rho_0^2)$ and $\mathcal{S}=\sum_n n \kedbrad{K_n}{K_n}$. Using the bound in Eq. \eqref{EMT} and the inequality $\frac{ C^2_{K}(t)}{4 \norm{\mathcal{S}}^2_{\rm op}}\leq 1 - \frac{\bradked{\rho_0}{\rho_t}^2}{(\tr \rho^2_0)^2}$ (see Appendix~\ref{appendixF}), we obtain
\begin{align}\label{CB}
   \arcsin( \frac{ C_{K}(T)}{2 \norm{\mathcal{S}}_{\rm op}}) \leq  \int_{0}^{T}\, dt \Delta \mathcal{L}_{\rm nc}.
\end{align}
Eq.~\eqref{SFFB} and Eq.~\eqref{CB} suggest that the growth of SFF and Krylov complexity of states are constrained by the speed of evolution of the quantum system in Liouville space. Furthermore, we can replace $\Delta \mathcal{L}_{\rm nc}$ by $\Delta \mathcal{L}$ in Eq.~\eqref{SFFB} and Eq.~\eqref{CB} to obtain loose bounds, however, $\Delta \mathcal{L}$ is time-independent in the case of unitary dynamics.
Moreover, for $\rho_0=\rho_\beta$ and unitary dynamics, the inequality mentioned above Eq.~\eqref{CB} provides a trade-off relationship between SFF and Krylov complexity of states:
\begin{equation}
\frac{ C^2_{K}(T)}{4 \norm{\mathcal{S}}^2_{\rm op}} + {\rm SFF^2}(T)\leq 1.
\end{equation}

\textbf{\it Explanation of Mpemba effect in non-equilibrium open quantum dynamics:---} The Mpemba effect says that far-from-equilibrium states may relax faster than states closer to equilibrium. Recently, it has been extensively studied in non-equilibrium open quantum dynamics~\cite{Chatterjee2023, Chatterjee2023a, Rylands2023, Strachan2024, Nava2024, Moroder2024}. Here, we show that the Mpemba effect can be understood using the speed efficiency. The speed efficiency is defined as 
\begin{equation}\label{speff}
    0 \leq \eta := \frac{\langle\!\langle\Delta \mathcal{L}\rangle\! \rangle_{T}}{\norm{\cal{L}}_{\rm op}} \leq 1.
\end{equation}
The above definition is based on the tightness of the inequality $\Delta \cal{L}\leq \norm{\cal{L}}_{\rm op}$, where the upper bound $\norm{\cal{L}}_{\rm op}$ can be interpreted as the maximal possible speed of evolution. Note that similar speed efficiency was defined for unitary evolution in Refs.~\cite{Uzdin2012, Rossetti2025}. In general, the relaxation timescale of a state depends on its speed efficiency, as higher speed efficiency implies larger Liouvillian fluctuations and consequently a faster evolution toward equilibrium. Therefore, a far-from-equilibrium state with higher speed efficiency can relax faster than a state closer to equilibrium but with lower speed efficiency. For example, consider two states $\rho_0'$ and $\rho_0''$ such that $\Theta(\rho_0', \rho_T) < \Theta(\rho_0'', \rho_T) $. If $\eta_{\rho_0'}< \eta_{\rho_0''}$, then even though $\rho_0''$ is far away from equilibrium compared to $\rho_0'$, it relaxes much faster. Moreover, if $\eta_{\rho_0''}$ is sufficiently larger that $\eta_{\rho_0'}$, then the minimal relaxation time scale may satisfy: $T_{QSL}''< T_{QSL}'$. It is important to note that for faster evolution the bound given in Eq.~\eqref{LMT} is likely to saturate.

In previous works, Mpemba effect was demonstrated using the eigenmodes of Liouvillion. Since, the minimal relaxation time-scale $T_{QSL}$ (given in Eq.~\eqref{LMT} ) is dictated by the Liouvillian eigenvalues and the overlap of initial state with left eigenvectors of Liouvillian (see Appendix~\ref{Exp}), this suggests a fundamental connection between the Mpemba effect and quantum speed limits.

Let us consider a qubit system interacting with a thermal bath at zero temperature. The dynamics of such a qubit system (bare Hamiltonian \(H=-\frac{\omega}{2}\sigma_z\)) is described by the following Lindblad master equation~\cite{Breuer2007,lidar2020} $ d\rho_{t}/{dt} = \gamma\left(\sigma_{-}\rho_t \sigma_{+}- \frac{1}{2}\{\sigma_{+}\sigma_{-},\rho_t\}\right),$ where $\sigma_{-} = \ketbra{0}{1}$ and $\sigma_{+}=\ketbra{1}{0}$ are the lowering and raising operators, and $\gamma$ is Weiskopf-Wigner decay constant. Let the initial state of the system be $\rho_0 = \ketbra{\psi}{\psi}$, where $\ket{\psi} = \alpha \ket{0}+\sqrt{1-\alpha^2}\ket{1}$. The state of the system at time $t$ is $\rho_t = (1+\left(\alpha ^2-1\right) e^{-t \gamma } )\ketbra{0}{0}+  e^{-\frac{1}{2}t \gamma} \alpha  \sqrt{1-\alpha ^2} (\ketbra{0}{1}+ \ketbra{1}{0})+ \left(1-\alpha ^2\right)e^{-t \gamma } \ketbra{1}{1}$. The steady state of the system is $\rho_{ss}=\ketbra{0}~(\alpha=1)$. The distance between an arbitrary pure initial state and the steady state is $\Theta(\rho_0,\rho_{ss})=\arccos(\alpha^2)$, which is a monotonically decreasing function of $\alpha$.  It is evident from Fig-\ref{Mpemba}(a) that for considered far away non-equilibrium states (small $\alpha$), the speed efficiency is higher than the closer non-equilibrium states (large $\alpha$). 
As a result, the distance $\Theta(\rho_{T},\rho_{ss})$ between the far away non-equilibrium states and the steady state decreases much faster than that between the closer non-equilibrium states and the steady state, as shown in Fig.-\ref{Mpemba}(b), which is the so-called Mpemba effect. Moreover, the bound given in Eq.~\eqref{LMT} is tighter for states with higher speed efficiency. This is quantified by the offset $\delta:=T-T_{QSL}$, where $T_{QSL}$ is the lower bound in Eq.~\eqref{LMT}, and is shown in Fig.-\ref{Mpemba}(c). For the mathematical expressions of the quantities needed to depict Fig.-\ref{Mpemba}, see Appendix~\ref{appendixG}, which also provides a demonstration of the Mpemba effect in the presence of a finite temperature bath.

{\it Discussion on the comparison with previously derived geometric speed limits:--} The previously derived geometric speed limits for arbitrary dynamics saturate for the dynamics that follow the corresponding geodesic~\cite{Taddei2013, Pires2016, Campaioli2019, Funo2019, Connor2021, Nakajima2022, Lan2022, Mai2023,  Rosal2024, Mai2024}. However, finding the general class of dynamics, i.e., Kraus operators or Liouvillian, which connects given initial and target states along a geodesic, remains unknown. The bound in Theorem~\ref{MT} also saturates for the dynamics that take the states along the geodesic. However, since the bound in Theorem \ref{MT} does not require calculating the square root of density matrices, Fisher information, or norms of non-Hermitian operators, it is arguably easier to compute than the previous speed limits. Moreover, the dynamics that saturates the bound in Theorem \ref{MT} is known (see Eq.~\eqref{dynsat}). We have also given a description of our bound in terms of the eigenmodes of the Liouvillian which provides a connection between the relaxation time scale of a system and the eigenmodes of the Liouvillian. Furthermore, the bound in Theorem~\ref{MT} can be further tightened by removing the classical contribution from the speed of evolution, which yields the lower bound in Theorem~\ref{ExactMT}. The lower bound in Theorem~\ref{ExactMT} is tighter than the bound given in Theorem~\ref{MT} since it does not overestimate the speed of evolution. Finally, the exact speed limit in Theorem~\ref{ExactMT}, which is derived by removing the classical part from the evolution speed and replacing the geodesic distance with the Wootters distance, is exact for all states and all dynamics. Hence, considering the above reasons, our bounds naturally outperform all previously obtained speed limits.

\textbf{\it Conclusions:---} In this work, we have addressed the problem of the attainability and computational complexity of quantum speed limits for completely positive and trace-preserving (CPTP) dynamics. We have derived exact and inexact quantum speed limits for CPTP dynamics for finite-dimensional quantum systems in Liouville space. The exact and inexact speed limits are a consequence of the exact and inexact uncertainty relations for operators acting on Liouville space, respectively. The inexact speed limit saturates for time-optimal (fastest) CPTP dynamics, whereas the exact speed limit is exact for all CPTP dynamics and all states. We have characterized the time-optimal CPTP dynamics and obtained its Liouvillian for given initial and orthogonal final states. We have provided the geometrical interpretation of the exact and inexact speed limits using the geometry of quantum evolution in Liouville space. We found that only the non-classical part of the Liouvillian is responsible for changes in the distinguishability of the initial and time-evolved states, which leads to the refining of the speed of evolution. Moreover, we showed that the new bounds can be regarded as the exact and inexact generalization of the Mandelstam-Tamm bound as they provide novel form of time-energy-dissipation uncertainty relations. 

We have provided two important applications of our bounds. First, we showed that the speed of evolution bounds the growth of spectral form factor and Krylov complexity of states, which can have useful applications in many-body and high-energy physics. Second, using the geometry of quantum evolution in Liouville space, particularly the tightness of inexact speed limits, we have successfully explained the counter-intuitive phenomenon of the Mpemba effect in open quantum dynamics. Our speed limits significantly enhance our understanding of the relationship between time, energy, and dissipation in open quantum dynamics and reveal the novel properties of CPTP dynamics. Since the exact and inexact quantum speed limits have been derived from fundamentally well-established concepts such as uncertainty relations, non-commutativity, and the geometry of quantum evolution, we believe they can be verified in existing experimental setups, as various functionals of density matrices can be measured using the SWAP test~\cite{Ekert02}. Additionally, our speed limits are simpler to compute, requiring estimating only the overlap of density matrices and the variance of the Liouvillian. We expect that our speed limits may have numerous applications in the rapidly growing field of quantum information science and technologies. The ideas presented in this article open the door to finding the exact or tight speed limits for other physical quantities of interest in many-body physics~\cite{Vikram2024}, high-energy physics~\cite{Jha24}, quantum thermodynamics~\cite{Mohan2022,Bera2022}, and other diverse areas.

\medskip
\textbf{\it Acknowledgments:---}
BM acknowledges Manabendra Nath Bera for fruitful discussions. BM acknowledge funding by the Research Council of Finland by Grant No.~355824.

\onecolumngrid
\section*{Appendix}
\appendix

Here, we include detailed preliminaries, derivations, and calculations to supplement the results presented in the main text.

\section{Preliminaries, Notations, Inequalities, and Equalities Used in the Main Text}\label{appendixA}

{\it \bf Liouville space:} For any quantum system of dimension \(d\), there exists an associated Hilbert space \(\mathsf{H}_d\). Let \(\mathbb{B}(\mathsf{H}_d)\) denote the set of all linear operators acting on \(\mathsf{H}_d\). There is an isomorphism between \(\mathbb{B}(\mathsf{H}_d)\) and \(\mathsf{H}_d \otimes \mathsf{H}_d\). Using this isomorphism, we can vectorize an arbitrary linear operator \(A \in \mathbb{B}(\mathsf{H}_d)\) as follows:
\begin{equation*}
    A = \sum_{ij} \alpha_{ij}\ketbra{i}{j}~~\rightarrow~~\ked{A} = \sum_{ij} \alpha_{ij} \ket{j} \otimes \ket{i}= \sum_{ij} \alpha_{ij} \ked{ji},
\end{equation*}
where \(\{\ket{i}\}\) forms an orthonormal basis for \(\mathsf{H}_d\). The space of vectorized linear operators \(\{\ked{A}\}\) is called Liouville space (doubled Hilbert space), which is of dimension $d^2$~\cite{Gyamfi2020}. The Liouville space is equipped with an inner product given by \(\bradked{A}{B} := \tr(A^{\dag}B)\). The state of a quantum system is represented by a linear positive semidefinite operator \(\rho\) with a unit trace called the density matrix. The vectorized form of any density matrix \(\rho\) is given by
\[
\rho = \sum_{ij} \rho_{ij}\ketbra{i}{j} \quad \rightarrow \quad \ked{\rho} = \sum_{ij} \rho_{ij} \ked{ji}.
\]
Note that \(\bradked{\rho}{\rho} = \tr(\rho^2) \neq 1\) in general, hence we can define the corresponding normalized vector as \(\ked{\tilde{\rho}} = \frac{\ked{\rho}}{\sqrt{\tr(\rho^2)}}\).\\

{\it \bf Liouville space angle:} The shortest distance between two density matrices, \(\rho_a\) and \(\rho_b\), in Liouville space is defined as follows:
\begin{equation}
    \Theta(\rho_a, \rho_b) := \arccos{\left(\tilde{\rho}_a | \tilde{\rho}_b\right)}, \hspace{0.3cm}\text{with}\hspace{0.3cm}
(\tilde{\rho}_a | \tilde{\rho}_b) := \frac{\tr(\rho_a \rho_b)}{\sqrt{\tr(\rho_a^2) \tr(\rho_b^2)}},
\end{equation}
where $ \Theta(\rho_a, \rho_b)$ is the Liouville space angle between state vectors $\ked{\tilde{\rho}_a}$ and $\ked{\tilde{\rho}_b}$. \(\Theta(\rho_a, \rho_b)\) can be thought of as an analog of the Hilbert space angle between state vectors. It satisfies all the properties required for a valid distance measure:
\begin{enumerate}
    \item Positive semidefinite: \(\Theta(\rho_a, \rho_b) \geq 0\) for all \(\rho_a\) and \(\rho_b\).
    \item Symmetric: \(\Theta(\rho_a, \rho_b)=\Theta(\rho_p, \rho_a)\) for all \(\rho_a\) and \(\rho_b\).
    \item Triangle inequality: \(\Theta(\rho_a, \rho_b) \leq \Theta(\rho_a, \rho_c) + \Theta(\rho_c, \rho_b)\) for all \(\rho_a\), \(\rho_b\) and $\rho_c$.
\end{enumerate}

{\it \bf Space of operators acting on Liouville space}: The linear operators that act on the vectors in Liouville space are denoted by \(\mathcal{A}, \mathcal{B}, \ldots\), etc. These operators acting on Liouville space form a linear space equipped with an inner product $\langle \mathcal{A}, \mathcal{B} \rangle:= \tr\left(\mathcal{A}^{\dag}\mathcal{B}\right)$. 
The expectation value and the variance of operator \(\mathcal{O}\) in the state vector \(\ked{\tilde{\rho}}\) is defined as:
\begin{align*}
(\!(\mathcal{O})\!) &:= \brad{\tilde{\rho}}\mathcal{O}\ked{\tilde{\rho}} = \tr(\mathcal{O}\mathcal{P}) \quad \text{and} \quad (\Delta \mathcal{O})^2 = \tr(\mathcal{O}^\dagger \mathcal{O}\mathcal{P}) - \tr(\mathcal{O}^\dagger\mathcal{P}) \tr(\mathcal{O}\mathcal{P}),
\end{align*}
where \(\mathcal{P} = \kedbrad{\tilde{\rho}}{\tilde{\rho}}\) is the projection operator associated with the state vector \(\ked{\tilde{\rho}}\). \\

{\it \bf Uncertainty relation for non-Hermitian operators acting on Liouville space} (Proposition 1): For any two non-Hermitian operators \(\mathcal{A}\) and \(\mathcal{B}\), the following uncertainty relation holds:
\begin{equation} \label{uncertaintyrelnsuper}
    (\Delta \mathcal{A})^2 (\Delta \mathcal{B})^2 \geq \left| \tr(\mathcal{A}^\dagger \mathcal{B}\mathcal{P}) - \tr(\mathcal{A}^\dagger\mathcal{P}) \tr(\mathcal{B}\mathcal{P}) \right|^2,
\end{equation}
where \((\Delta \mathcal{O})^2 = \tr(\mathcal{O}^\dagger \mathcal{O}\mathcal{P}) - \tr(\mathcal{O}^\dagger\mathcal{P}) \tr(\mathcal{O}\mathcal{P})\) is the variance of the operator \(\mathcal{O}\) in the state vector \(\ked{\tilde{\rho}}\).
\begin{proof}
    Let us define two vectors in Liouville space as \(\ked{f} := (\mathcal{A} - \tr(\mathcal{A}\mathcal{P}))\ked{\tilde{\rho}}\) and \(\ked{g} := (\mathcal{B} - \tr(\mathcal{B}\mathcal{P}))\ked{\tilde{\rho}}\), where \(\mathcal{A}\) and \(\mathcal{B}\) are operators acting on Liouville space, and \(\ked{\tilde{\rho}}\) is a normalized vector in Liouville space. Note that the inner product between the vectors \(\ked{A}\) and \(\ked{B}\) in Liouville space is defined as \(\bradked{A}{B} := \tr{(A^{\dag}B)}\). Thus, the inequalities that hold for inner product spaces must also exist in Liouville space. Utilizing the Cauchy-Schwarz inequality, we obtain
    \begin{equation}
        \bradked{f}{f} \bradked{g}{g} \geq \abs{\bradked{f}{g}}^2.
    \end{equation}
    Now, by the definitions of \(\ked{f}\) and \(\ked{g}\), we have \(\bradked{f}{f} = \Delta \mathcal{A}^2\) and \(\bradked{g}{g} = \Delta \mathcal{B}^2\). Using these relations in the above inequality, we obtain \((\Delta \mathcal{A})^2 (\Delta \mathcal{B})^2 \geq \abs{\tr(\mathcal{A}^\dagger \mathcal{B}\mathcal{P}) - \tr(\mathcal{A}^\dagger\mathcal{P}) \tr(\mathcal{B}\mathcal{P})}^2\), which completes the proof.
\end{proof}
Important to note that in Eq.~\eqref{uncertaintyrelnsuper}, if we replace the operators acting on Liouville space $\mathcal{A}$, $\mathcal{B}$, $\mathcal{P}$ by operators acting on Hilbert space $A$, $B$, $\rho$ (density matrix), respectively, we retrieve the uncertainty relation for non-Hermitian operators obtained in the Ref.~\cite{Pati2015}, which has recently been experimentally verified~\cite{Zhao2024}.\\

{\it \bf Exact uncertainty relation for operators acting on Liouville space (Proposition 2):} For any Hermitian operator $\mathcal{A}=\sum_i a_{i} \kedbrad{a_{i}}{a_{i}}$ and a non-Hermitian operator $\mathcal{B}$ there exists an exact uncertainty relation for all state vectors  $\ked{\tilde{\rho}}$ as
     \begin{equation}\label{exactuncrln1}
         \delta_{\mathcal{B}}\mathcal{A} \Delta \mathcal{B}_{\rm nc}=1/2,
     \end{equation}
     where
    \begin{equation*}
    (\delta_{\mathcal{B}}\mathcal{A})^{-2}:=\sum_{i}\frac{ \brad{a_{i}}(\mathcal{B}\mathcal{P}+\mathcal{P} \mathcal{B}^\dagger-\mathcal{P}\tr[\mathcal{P}(\mathcal{B}+\mathcal{B}^\dagger)])\ked{a_{i}}^2}{(a_{i}|\mathcal{P}|a_{i})},~~\mathcal{B}=\mathcal{B}_{\rm cl}+\mathcal{B}_{\rm nc}~~\text{and}~~\mathcal{B}_{\rm cl}:=\sum_{i} \kedbrad{a_{i}}{a_{i}}\frac{(a_{i}|\frac{1}{2}(\mathcal{B}\mathcal{P}-\mathcal{P} \mathcal{B}^\dagger)|a_{i})}{(a_{i}|\mathcal{P}|a_{i})}.
\end{equation*}

\begin{proof}
Given a Hermitian operator $\mathcal{A}$, we can always decompose a non-Hermitian operator $\mathcal{B}$ as a sum of two operators as $\mathcal{B}=\mathcal{B}_{\rm cl}+\mathcal{B}_{\rm nc}$, where $\mathcal{B}_{\rm cl}$ is anti-Hermitian $(\mathcal{B}_{\rm cl}^\dagger=-\mathcal{B}_{\rm cl})$ and is diagonal in the eigenbasis of $\mathcal{A}$, and $\mathcal{B}_{\rm nc}$ is a non-Hermitian operator. The operator $\mathcal{B}_{\rm cl}$ is defined as
\begin{equation}\label{classicalop}
   \mathcal{B}_{\rm cl}:=\sum_{i} \kedbrad{a_{i}}{a_{i}}\frac{(a_{i}|\frac{1}{2}(\mathcal{B}\mathcal{P}-\mathcal{P} \mathcal{B}^\dagger)|a_{i})}{(a_{i}|\mathcal{P}|a_{i})},
\end{equation}
where $\{\ked{a_i}\}$ is the eigenbasis of $\mathcal{A}$. Then the variance of the non-classical part $\mathcal{B}_{\rm nc}$ in the state $\ked{\tilde{\rho}}$ is given by
\begin{eqnarray}\label{varBnc}
    \Delta \mathcal{B}_{\rm nc}^2&=&\Delta (\mathcal{B}-\mathcal{B}_{\rm cl})^2\nonumber\\
    &=&\tr((\mathcal{B}-\mathcal{B}_{\rm cl})^\dagger(\mathcal{B}-\mathcal{B}_{\rm cl})\mathcal{P})-\tr((\mathcal{B}-\mathcal{B}_{\rm cl})^\dagger\mathcal{P})\tr((\mathcal{B}-\mathcal{B}_{\rm cl})\mathcal{P}).
\end{eqnarray}
Let us consider the first term of the R.H.S. of the above equation
\begin{eqnarray}\label{firstterm}
    \tr((\mathcal{B}-\mathcal{B}_{\rm cl})^\dagger(\mathcal{B}-\mathcal{B}_{\rm cl})\mathcal{P})&=&\tr(\mathcal{B}^\dagger B\mathcal{P}-\mathcal{B}_{\rm cl}^\dagger B\mathcal{P}- \mathcal{B}^{\dagger}\mathcal{B}_{\rm cl}\mathcal{P}+\mathcal{B}_{\rm cl}^\dagger\mathcal{B}_{\rm cl}\mathcal{P})+\tr(\mathcal{B}_{\rm cl}^\dagger\mathcal{B}_{\rm cl}\mathcal{P})-\tr(\mathcal{B}_{\rm cl}^\dagger\mathcal{B}_{\rm cl}\mathcal{P})\nonumber\\
    &=&\tr(\mathcal{B}^\dagger B\mathcal{P})-\tr(\mathcal{B}_{\rm cl}^\dagger\mathcal{B}_{\rm cl}\mathcal{P})+2\tr(\mathcal{B}_{\rm cl}^\dagger\mathcal{B}_{\rm cl}\mathcal{P})-\tr(\mathcal{B}_{\rm cl}\mathcal{P}\mathcal{B}^\dagger+\mathcal{B}_{\rm cl}^\dagger\mathcal{B}\mathcal{P})\nonumber\\
    &=&\tr(\mathcal{B}^\dagger B\mathcal{P})-\tr(\mathcal{B}_{\rm cl}^\dagger\mathcal{B}_{\rm cl}\mathcal{P})-\sum_{i}\frac{\bradked{a_{i}}{(\mathcal{B}\mathcal{P}-\mathcal{P}\mathcal{B}^\dagger)|a_{i}}^2}{2}+\sum_{i}\frac{\bradked{a_{i}}{(\mathcal{B}\mathcal{P}-\mathcal{P}\mathcal{B}^\dagger)|a_{i}}^2}{2}\nonumber\\
    &=&\tr(\mathcal{B}^\dagger B\mathcal{P})-\tr(\mathcal{B}_{\rm cl}^\dagger\mathcal{B}_{\rm cl}\mathcal{P}),
\end{eqnarray}
where to get the third equality, we used the definition of $\mathcal{B}_{\rm cl}$ from Eq.~\eqref{classicalop}. Similarly, we can simplify the second term in Eq.~\eqref{varBnc}
\begin{eqnarray}\label{secondterm}
    \tr((\mathcal{B}-\mathcal{B}_{\rm cl})^\dagger\mathcal{P})\tr((\mathcal{B}-\mathcal{B}_{\rm cl})\mathcal{P})&=&\tr(\mathcal{B}^\dagger\mathcal{P})\tr(\mathcal{B}\mathcal{P})-\tr(\mathcal{B}^\dagger\mathcal{P})\tr(\mathcal{B}_{\rm cl}\mathcal{P})-\tr(\mathcal{B}_{\rm cl}^\dagger\mathcal{P})\tr(\mathcal{B}\mathcal{P})+\tr(\mathcal{B}_{\rm cl}^\dagger\mathcal{P})\tr(\mathcal{B}_{\rm cl}\mathcal{P})\nonumber\\
    &=&\tr(\mathcal{B}^\dagger\mathcal{P})\tr(\mathcal{B}\mathcal{P})-\tr(\mathcal{B}_{\rm cl}^\dagger\mathcal{P})\tr(\mathcal{B}_{\rm cl}\mathcal{P})+\tr(\mathcal{B}_{\rm cl}\mathcal{P})\left(-2\tr(\mathcal{B}_{\rm cl}\mathcal{P})-\tr(\mathcal{B}^\dagger\mathcal{P})+\tr(\mathcal{B}\mathcal{P})\right)\nonumber\\
    &=&\tr(\mathcal{B}^\dagger\mathcal{P})\tr(\mathcal{B}\mathcal{P})-\tr(\mathcal{B}_{\rm cl}^\dagger\mathcal{P})\tr(\mathcal{B}_{\rm cl}\mathcal{P}),
\end{eqnarray}
where in the second equality we have added and subtracted $\tr(\mathcal{B}_{\rm cl}^\dagger\mathcal{P})\tr(\mathcal{B}_{\rm cl}\mathcal{P})$ on the R.H.S. of the equation and used the anti-Hermiticity $\mathcal{B}_{\rm cl}^\dagger=-\mathcal{B}_{\rm cl}$ to simplify its last term. Finally to get the last equality we have used the fact that $\tr\mathcal{B}_{\rm cl}(\mathcal{P})=1/2(\tr(\mathcal{B}\mathcal{P})-\tr(\mathcal{B}^\dagger\mathcal{P}))$ (see Eq.~\eqref{classicalop}). Thus using Eqs.~\eqref{firstterm} and \eqref{secondterm} in Eq.~\eqref{varBnc}, we obtain
\begin{eqnarray} \label{varianceB}
    \Delta \mathcal{B}_{\rm nc}^2&=&\Delta \mathcal{B}^{2}-\Delta \mathcal{B}_{\rm cl}^{2}\nonumber\\
    &=&\tr(\mathcal{B}^{\dagger} \mathcal{B}\mathcal{P})-\tr(\mathcal{B}^{\dagger}\mathcal{P})\tr(\mathcal{B}\mathcal{P})-\tr( \mathcal{B}_{\rm cl}^{\dagger}\mathcal{B}_{\rm cl}\mathcal{P})+\tr(\mathcal{B}_{\rm cl}^{\dagger}\mathcal{P})\tr(\mathcal{B}_{\rm cl}\mathcal{P})\nonumber\\
    &=&\frac{1}{4}\left[\sum_{i}\left(\brad{a_i}\mathcal{B}\ked{\tilde{\rho}}^2+\brad{\tilde{\rho}}\mathcal{B}^{\dagger}\ked{a_i}^2\right)-\left(\tr(\mathcal{B}\mathcal{P})+\tr( \mathcal{B}^\dagger\mathcal{P})\right)^2+2\tr( \mathcal{B}^{\dagger}\mathcal{B}\mathcal{P})\right].
    \end{eqnarray}
Let us now define a measure of uncertainty for the operator $\mathcal{A}$, denoted by $\delta_{\mathcal{B}}\mathcal{A}$, as
\begin{align}
 (\delta_{\mathcal{B}}\mathcal{A})^{-2}=\sum_i\frac{ \brad{a_i}(\mathcal{B}\mathcal{P}+\mathcal{P} \mathcal{B}^\dagger-\mathcal{P}\tr[\mathcal{P}(\mathcal{B}+\mathcal{B}^\dagger)])\ked{a_i}^2}{(a_i|\mathcal{P}|a_i)}.\nonumber
 \end{align}
 For $\mathcal{P}=\kedbrad{\tilde{\rho}}{\tilde{\rho}}$, the above equation can be rewritten as
 \begin{align} \label{uncertaintyA}
(\delta_{\mathcal{B}}\mathcal{A})^{-2}=\left[\sum_i\left(\brad{a_i}\mathcal{B}\ked{\tilde{\rho}}^2+\brad{\tilde{\rho}}\mathcal{B}^{\dagger}\ked{a_i}^2\right)-\left(\tr(\mathcal{B}\mathcal{P})+\tr( \mathcal{B}^\dagger\mathcal{P})\right)^2+2\tr( \mathcal{B}^{\dagger}\mathcal{B}\mathcal{P})\right].
\end{align}
Finally, using Eqs.~\eqref{varianceB} and \eqref{uncertaintyA}, we obtain
\begin{equation}
         \delta_{\mathcal{B}}\mathcal{A} \Delta \mathcal{B}_{\rm nc}=1/2,
     \end{equation}
which holds for an arbitrary state vector $\ked{\tilde{\rho}}$ corresponding to the density matrix $\rho$. It is important to note that $\mathcal{B}_{\rm cl}$ commutes with $\mathcal{A}$, whereas $\mathcal{B}_{\rm nc}$ does not commute with $\mathcal{A}$, and the exact uncertainty relation is based on the non-commutativity of the given two operators acting on Liouville space.
\end{proof}
Interestingly, in Eq.~\eqref{exactuncrln1}, if we replace the operators acting on Liouville space $\mathcal{A}$, $\mathcal{B}$, and $\mathcal{P}$ with the operators acting on Hilbert space $A$, $B$, and $\rho$ (density matrix), respectively, we obtain the exact uncertainty relation for two operators $A$ and $B$ (where $A$ is Hermitian and $B$ is non-Hermitian). This result reduces to the exact uncertainty relation obtained by {\it M.J. Hall} in Ref.~\cite{Hall2001} if both operators are Hermitian and $\rho$ is a pure state.\\

{\it \bf Dynamical equation for CPTP dynamics:}
Here, we consider the physical processes, which are described by completely positive and trace-preserving (CPTP) dynamics. For finite-dimensional systems, any CPTP dynamics can be described by a time-local master equation, which can be written as~\cite{Lindblad1976, Gorini1976, Breuer2007}:
\begin{equation}\label{mastereq}
\frac{d\rho_t}{dt}  = \mathbb{L}_{t}(\rho_t) =-i[H_t,\rho_t] + \sum_{k=1}^{d^2-1}\gamma^{k}_{t}\left(L^{k}_{t}\rho_t L^{k\dagger}_{t}-\frac{1}{2}\{L^{k\dagger }_{t}L^{k}_{t},\rho_t\}\right),
\end{equation}
where \(H_t\) is the driving Hamiltonian of the system and \(L^{k}_t\) are the Lindblad operators with rates \(\gamma^{k}_t\). If \(\gamma^{k}_t\geq0\) the corresponding dynamical map 
\(\Lambda_t = \mathcal{T} \exp \left( \int_{0}^{t} d\tau \mathbb L_{\tau} \right)\) is completely positive and divisible, where $\mathcal{T}$ is the time ordering operator. However, when the Hamiltonian, Lindblad operators, and the rates are time-independent, the map can be written as  \(\Lambda_t =  \exp \left( \mathbb L t\right)\)~\cite{Breuer2007, Dariusz2010, Settimo2024}. It is important to note that there may exist a CPTP dynamical map $\Lambda_t$ that cannot be expressed in the form of a Lindblad master equation. 

Recall that for the product of three operators $A,B$ and $C$ the following identity holds~\cite{Gyamfi2020}:
 \begin{equation}\label{tripleprodidentity}
    \ked{ABC}=(C^T\otimes A)\ked{B},
\end{equation}
where $C^{T}$ is the transpose of $C$. 

The vectorized form of the above master equation can be written as
\begin{equation}\label{eomforrho}
    \frac{d}{dt}\ked{\rho_t}=\mathcal{L}_{t}\ked{\rho_t},
\end{equation}
where $\mathcal{L}_t$ is called Liouvillian. Note that $\mathcal{L}_t$ is a super-matrix associated with the Lindblad superoperator $\mathbb{L}_t$~\cite{Sarandy_2004}, which can be written using Eqs.~\eqref{tripleprodidentity} and \eqref{mastereq} as follows
\begin{eqnarray}\label{genofdynvec}
 \ked{[H_{t},\rho_t]}&=&\ked{H_t\rho_t\mathbb{I}}-\ked{\mathbb{I}\rho_t H_{t}}= \left(\mathbb{I}\otimes H_t-H_{t}^T\otimes\mathbb{I}\right)\ked{\rho_t},\nonumber\\
    \ked{L^{k}_{t}\rho_t L^{k\dagger}_{t}}-\frac{1}{2}\ked{\{L^{k\dagger }_{t}L^{k}_{t},\rho_t\}}&=&\left(L^{k*}_{t}\otimes L^{k}_t -\frac{1}{2}\left(\mathbb{I}\otimes L^{k\dagger}_{t}L^{k}_{t}+(L^{k\dagger }_{t}L^{k}_{t})^T\otimes\mathbb{I}\right)\right)\ked{\rho_t},
\end{eqnarray}
where $\mathbb{I}$ is the identity operator. 
Using the above two equations, we obtain
\begin{eqnarray} \label{decompofL}
    \mathcal{L}_{t} &=& -i\mathcal{L}_H+\mathcal{L}_D,
\end{eqnarray}
where $\mathcal{L}_H :=\left(\mathbb{I}\otimes H_t-H_{t}^T\otimes\mathbb{I}\right)$ is Hermitian and is the generator of the unitary dynamics, and
\begin{equation}
    \mathcal{L}_D :=\sum_{k} \gamma_{k} \left(L^{k*}_{t}\otimes L^{k}_t -\frac{1}{2}\left(\mathbb{I}\otimes L^{k\dagger}_{t}L^{k}_{t}+(L^{k\dagger }_{t}L^{k}_{t})^T\otimes\mathbb{I}\right)\right)
\end{equation}  is non-Hermitian and is the generator of the dissipative dynamics.

The above master Eq.~\eqref{eomforrho} does not preserve the norm of the state; therefore, the norm preserving master equation for any normalized state $\ked{\tilde\rho_t}$ can be written as 
\begin{eqnarray}\label{Master}
    \frac{d}{dt}\ked{\tilde\rho_t}&=&\frac{d}{dt}\left(\frac{\ked{\rho_t}}{\sqrt{\bradked{\rho_t}{\rho_t}}}\right)
    = \frac{d}{dt}\ked{{\tilde\rho_t}}=(\mathcal{L}_t-\frac{1}{2}\{(\tilde\rho_t|\mathcal{L}_t|\tilde\rho_t)+(\tilde\rho_t|\mathcal{L}_t^\dagger|\tilde\rho_t)\})\ked{{\tilde\rho_t}}, \label{equ:arbitraryL}
\end{eqnarray}
where $\ked{\tilde\rho_t}=\frac{\ked{\rho_t}}{\sqrt{\bradked{\rho_t}{\rho_t}}}$ and $\bradked{\rho_t}{\rho_t}=\tr(\rho_t^2)$. The above equation is valid for arbitrary Liouvillian. However, for the Lindblad master equation, we can further simplify the above equation by using the decomposition of $\mathcal{L}_t$ as $\mathcal{L}_t=-i\mathcal{L}_H+\mathcal{L}_D$, we have
\begin{eqnarray}\label{tr(LP)isrealproof}
    (\tilde\rho_t|(\mathcal{L}_t-\mathcal{L}_t^\dagger)|\tilde\rho_t)&=&(\tilde\rho_t|(-2i\mathcal{L}_H+\mathcal{L}_D-\mathcal{L}_D^\dagger)|\tilde\rho_t)\nonumber\\
    &=&-2i\bradked{\tilde{\rho}_t}{\mathcal{L}_H|\tilde{\rho}_t}+\bradked{\tilde{\rho}_t}{\mathcal{L}_D|\tilde{\rho}_t}-\bradked{\tilde{\rho}_t}{\mathcal{L}_D^{\dagger}|\tilde{\rho}_t}\nonumber\\
    &=&-2i(\tilde{\rho}_t|\left(\mathbb{I}\otimes H_t-H_t^T\otimes\mathbb{I}\right)|\tilde{\rho}_t)+\sum_k\gamma^k\bradked{\tilde{\rho}_t}{\left(L^{k*}_{t}\otimes L^{k}_t -\frac{1}{2}\left(\mathbb{I}\otimes L^{k\dagger}_{t}L^{k}_{t}+(L^{k\dagger }_{t}L^{k}_{t})^T\otimes\mathbb{I}\right)\right)|\tilde{\rho}_t}\nonumber\\
    &&-\sum_k\gamma^k\bradked{\tilde{\rho}_t}{\left(L^{k^{T}}_{t}\otimes L^{k\dagger}_t -\frac{1}{2}\left(\mathbb{I}\otimes L^{k\dagger}_{t}L^{k}_{t}+(L^{k\dagger }_{t}L^{k}_{t})^{*}\otimes\mathbb{I}\right)\right)|\tilde{\rho}_t}\nonumber\\
    &=&-2i\bradked{\tilde{\rho}_t}{[H_t,\tilde{\rho}_t]}+\sum_k \gamma^k \bradked{\tilde{\rho}_t}{L_t^k\tilde{\rho}_tL_t^{k\dagger}-\frac{1}{2}\{L_t^{k\dagger}L_t^k,\tilde{\rho}_t\}}-\sum_k \gamma^k \bradked{\tilde{\rho}_t}{L_t^{k\dagger}\tilde{\rho}_tL_t^{k}-\frac{1}{2}\{L_t^{k\dagger}L_t^k,\tilde{\rho}_t\}}\nonumber\\
    &=&-2i\tr(\tilde{\rho}_t[H_t,\tilde{\rho}_t])+\sum_k\gamma^k\left(\tr(\tilde{\rho}_tL_t^k\tilde{\rho}_tL_t^{k\dagger})-\tr(\tilde{\rho}_tL_t^{k\dagger}\tilde{\rho}_tL_t^k)\right)\nonumber\\
    &=&0,
\end{eqnarray}
where the fourth equality is due to Eq.~\eqref{genofdynvec} and in the fifth equality, we have used the fact that $\bradked{A}{B}=\tr(A^\dagger B)$. Finally, using the cyclic property of trace, we obtain the last equality. Hence, $\bradked{\tilde{\rho}_t}{\mathcal{L}|\tilde{\rho}_t}$ is a real number for any state $\ked{\tilde{\rho}_t}$, and thus for Lindblad master equation, Eq.~\eqref{equ:arbitraryL} can be written as
\begin{equation}\label{eomfornormrho}
    \frac{d}{dt}\ked{{\tilde\rho_t}}=\left(\mathcal{L}_t-\bradked{\tilde{\rho}_t}{\mathcal{L}_t|\tilde{\rho}_t}\right)\ked{\tilde{\rho}_t}.
\end{equation}
Using Eq.~\eqref{equ:arbitraryL}, the time evolution of the projection operator acting on Liouville space, ${\mathcal{P}}_t=\kedbrad{\tilde\rho_t}{\tilde\rho_t}$ associated with the time-evolved state vector $\ked{\tilde\rho_t}$, is given by the following master equation
\begin{eqnarray}\label{eomforprojector}
    \frac{d}{dt}\mathcal{P}_t=\mathcal{L}_{t}\mathcal{P}_t+\mathcal{P}_t\mathcal{L}_{t}^\dagger-\mathcal{P}_t\tr\left[(\mathcal{L}_{t}+\mathcal{L}_{t}^\dagger)\mathcal{P}_t\right].
\end{eqnarray}
 The Liouvillian can also be decomposed into an anti-Hermitian and a Hermitian part as \(\mathcal{L} = -i\mathcal{L}_{+} + \mathcal{L}_{-}\), where \(\mathcal{L}_{+}=\mathcal{L}_{H} +i\left(\mathcal{L}_{D}-\mathcal{L}^\dagger_{D}\right)/2\) and \(\mathcal{L}_{-}=\left(\mathcal{L}_{D}+\mathcal{L}^\dagger_{D}\right)/2\) are associated with the reversible and irreversible parts of the dynamics, respectively.

\section{Proof of Theorem~1}\label{appendixB}

Using Eq.~\eqref{eomforprojector}, the rate of change of the distinguishability between the initial state vector $\ked{\tilde{\rho}_{0}}$ and the time-evolved state vector $\ked{\tilde{\rho}_{t}}$ is given by 
\begin{eqnarray}\label{Dist}
\frac{d}{dt} (\tilde{\rho}_0|\tilde{\rho}_{t})^2&=&     \frac{d}{dt}\tr(\mathcal{P}_t \mathcal{P}_0)=\tr\left(\mathcal{L}\mathcal{P}_t\mathcal{P}_0\right)+\tr\left(\mathcal{P}_t\mathcal{L}^\dagger\mathcal{P}_0\right)-\tr\left(\mathcal{P}_t\mathcal{P}_0\right)\tr\left(\left(\mathcal{L}+\mathcal{L}^\dagger\right)\mathcal{P}_t\right)\label{rateofchangeofdis}\\
&=&X+X^*,
\end{eqnarray}
where $\mathcal{P}_t=\kedbrad{\tilde{\rho}_t}{\tilde{\rho}_t}$ and $\mathcal{P}_0=\kedbrad{\tilde{\rho}_0}{\tilde{\rho}_0}$ are the projection operators associated with the initial and time-evolved state vectors, respectively, $X:=\tr(\mathcal{L}\mathcal{P}_t \mathcal{P}_0)-\tr(\mathcal{P}_t\mathcal{P}_0)\tr(\mathcal{L}\mathcal{P}_t)$ and $X^*$ is the complex conjugate of $X$. Let us take the absolute value of both sides of the above equation, then we obtain
\begin{eqnarray}
    \left|\frac{d}{dt}\tr(\mathcal{P}_t \mathcal{P}_0)\right| &=&|X+X^*|\nonumber\\
    &\leq&2|X|\nonumber\\
    &=&2|\tr(\mathcal{L}\mathcal{P}_t \mathcal{P}_0)-\tr(\mathcal{P}_t\mathcal{P}_0)\tr(\mathcal{L}\mathcal{P}_t)|,\label{triangine}
\end{eqnarray}
where in the second step, we use the triangle inequality and the fact that $|X|=|X^*|$.

Using the uncertainty relation for non-Hermitian operators acting on Liouville space~\ref{uncertaintyrelnsuper} ($\mathcal{A}=\mathcal{P}_0$ and $\mathcal{B}=\mathcal{L}$) and the above inequality, we obtain the following bound on the rate of change of the distinguishability 
\begin{eqnarray}
    \left|\frac{d}{dt}\tr(\mathcal{P}_t \mathcal{P}_0)\right| 
    \leq2|\tr(\mathcal{L}\mathcal{P}_t \mathcal{P}_0)-\tr(\mathcal{P}_t\mathcal{P}_0)\tr(\mathcal{L}\mathcal{P}_t)|
    \leq2\Delta\mathcal{L}\Delta\mathcal{P}_0, \label{equ:rateofdis}
\end{eqnarray}
where  $(\Delta \mathcal{O})^2=\tr(\mathcal{O}^\dagger \mathcal{O}\mathcal{P}_{t})-\tr(\mathcal{O}^\dagger\mathcal{P}_{t}) \tr(\mathcal{O}\mathcal{P}_{t})$ is the variance of the operator $\mathcal{O}$ in the state $\ked{\tilde{\rho}_{t}}$. Let us integrate the above inequality on both sides and use the fact that $\abs{\int f(t)dt}\leq\int \abs{f(t)}dt$, then we obtain
\begin{equation}
\abs{\int_{0}^{T}\frac{d\tr(\mathcal{P}_t\mathcal{P}_0)}{\sqrt{\tr(\mathcal{P}_t\mathcal{P}_0)}\sqrt{1-\tr(\mathcal{P}_t\mathcal{P}_0)}}}\leq2\int_{0}^{T}\Delta\mathcal{L}dt.
\end{equation}
On performing the above integration, we obtain the following lower bound on the evolution time
\begin{equation} \label{mainbound}
T \geq \frac{ \Theta(\rho_0,\rho_T)}{\langle\!\langle\Delta\mathcal{L}\rangle\!\rangle_{T}},
\end{equation}
where ${\Theta(\rho_0,\rho_{T})}=\arccos\bradked{\tilde{\rho}_{0}}{\tilde{\rho}_{T}}$ is the Liouville space angle between the initial state vector $\ked{\tilde{\rho}_{0}}$ and the final state vector $\ked{\tilde{\rho}_{T}}$, and ${\langle\!\langle\Delta\mathcal{L}\rangle\!\rangle_{T}}=\frac{1}{T}\int_0^T \, dt \Delta \mathcal{L}$. 

Furthermore, we can write the variance of the Liouvillian as follows 
\begin{eqnarray}\label{Ln}
    \Delta\mathcal{L}^2&=&{\tr(\mathcal{L}^\dagger\mathcal{L}\mathcal{P}_t)-\tr(\mathcal{L}\mathcal{P}_t)^2}\nonumber\\
    &=&{\tr\left((i\mathcal{L}_H+\mathcal{L}_D^\dagger)(-i\mathcal{L}_H+\mathcal{L}_D)\mathcal{P}_t\right)-\tr\Big((-i\mathcal{L}_H+\mathcal{L}_D)\mathcal{P}_t\Big)^2}\nonumber\\
    &=&{\tr(\mathcal{L}_H^2\mathcal{P}_t)+i\tr\left((\mathcal{L}_H\mathcal{L}_D-\mathcal{L}_D^\dagger\mathcal{L}_H)\mathcal{P}_t\right)+\tr(\mathcal{L}_D^\dagger\mathcal{L}_D\mathcal{P}_t)-\tr(\mathcal{L}_D\mathcal{P}_t)^2}\nonumber\\
    &=&{\Delta\mathcal{L}_H^2+\Delta\mathcal{L}_D^2+i\tr\left((\mathcal{L}_H\mathcal{L}_D-\mathcal{L}_D^\dagger\mathcal{L}_H)\mathcal{P}_t\right)},
\end{eqnarray}
where in the second equality, we have used the decomposition $\mathcal{L}=-i\mathcal{L}_H+\mathcal{L}_D$, and the third equality is due to the fact that $\tr(\mathcal{L}_H\mathcal{P}_t)=0$ (see Eq.~\eqref{tr(LP)isrealproof}). Here, $\Delta \mathcal{L}_H^2 = \frac{2}{\tr{(\rho_t^2)}} \left(\tr\left(\rho_t^2 H_t^2\right)-\tr\left(\rho_t H_t\rho_t H_t\right)\right)$, which reduces to the variance of the Hamiltonian if $\rho_t$ is a pure state. It can be shown that $\Delta \mathcal{L}_H^2 \leq \frac{4}{\tr{(\rho_t^2)}}\Delta H_t$, where $\Delta H_t$ is variance of the Hamiltonian in state $\rho_t$~\cite{Luo2020}. Moreover, the last term of the above Eq.~\eqref{Ln} is upper bounded by $ 2 \Delta \mathcal{L}_{H} \Delta \mathcal{L}_{D}$. To see this, let us define two vectors $\ked{a} := \left(\mathcal{L}_{H}-\tr\left(\mathcal{L}_H \mathcal{P}_t\right)\right)\ked{\tilde{\rho}_t}$ and $\ked{b} := \left(\mathcal{L}_{D}-\tr\left(\mathcal{L}_D \mathcal{P}_t\right)\right)\ked{\tilde{\rho}_t}$. Now using the Cauchy-Schwarz inequality for $\ked{a}$ and $\ked{b}$, we get: 
\begin{align}
    \bradked{a}{a} \bradked{b}{b} \geq \abs{\bradked{a}{b}}^2 \geq \left(\frac{1}{2 i}\left[\bradked{a}{b}-\bradked{b}{a}\right]\right)^2,
\end{align}
where in the last inequality, we have used the fact that for any complex number $z$, we have $\abs{z}^2\geq (\rm{Im}(z))^2$. Now, by the definitions of \(\ked{a}\) and \(\ked{b}\), we have \(\bradked{a}{a} = \Delta \mathcal{L}^2_{H}\), \(\bradked{g}{g} = \Delta \mathcal{L}^2_{D}\), and $\bradked{a}{b}-\bradked{b}{a} =\tr\left((\mathcal{L}_H\mathcal{L}_D-\mathcal{L}_D^\dagger\mathcal{L}_H)\mathcal{P}_t\right) $. Using these in the above equation, we obtain
\begin{equation}
    2 \Delta \mathcal{L}_{H} \Delta \mathcal{L}_{D} \geq i \tr\left((\mathcal{L}_H\mathcal{L}_D-\mathcal{L}_D^\dagger\mathcal{L}_H)\mathcal{P}_t\right).
\end{equation}
If we take the decomposition of the Liouvillian into an anti-Hermitian and a Hermitian part as \(\mathcal{L} = -i\mathcal{L}_{+} + \mathcal{L}_{-}\), then we can write the variance of the Liouvillian as follows 
\begin{eqnarray}
    \Delta\mathcal{L}^2 &=&{\Delta\mathcal{L}_{+}^2+\Delta\mathcal{L}_{-}^2+i\tr\left(\left[\mathcal{L}_{+},\mathcal{L}_{-}\right]\mathcal{P}_t\right)}. 
\end{eqnarray}
\section{Geometrical derivation of Inexact Quantum Speed Limit for arbitrary time-continuous dynamics}\label{appendixC}
The Hilbert-Schmidt distance between two state vectors $\ked{\tilde{\rho}_1}$ and $\ked{\tilde{\rho}_2}$ is given as
\begin{equation}
    S^2 =\frac{1}{2}\norm{\mathcal{P}_1-\mathcal{P}_2}^2_{\rm HS} =(1- \bradked{\tilde{\rho}_1}{\tilde{\rho}_2}\bradked{\tilde{\rho}_2}{\tilde{\rho}_1}),
\end{equation}
where $\mathcal{P}_1 = \kedbrad{\tilde{\rho}_1}{\tilde{\rho}_1}$ and $\mathcal{P}_2 = \kedbrad{\tilde{\rho}_2}{\tilde{\rho}_2}$ are projection operators acting on Liouville space, associated with the state vectors $\ked{\tilde{\rho}_1}$ and $\ked{\tilde{\rho}_2}$, respectively.

Let us consider a quantum system whose evolution is governed by arbitrary time-continuous dynamics. At time $t$, its state vector is given by $\ked{\tilde\rho_t}$. The distance between two infinitesimally separated state vectors $\ked{\tilde\rho_t}$ and $\ked{\tilde\rho_{t+dt}}$ is given as
\begin{eqnarray}\label{metric}
    dS^{2}&=&(1- \bradked{\tilde{\rho}_t}{\tilde{\rho}_{t+dt}}\bradked{\tilde{\rho}_{t+dt}}{\tilde{\rho}_t})\nonumber\\
    &=&\left(1-\frac{\bradked{\rho_t}{\rho_{t+dt}}\bradked{\rho_{t+dt}}{\rho_t}}{\bradked{\rho_t}{\rho_t}\bradked{\rho_{t+dt}}{\rho_{t+dt}}}\right),
\end{eqnarray}
where $\ked{\tilde\rho_t}=\ked{\rho_t}/\sqrt{\bradked{\rho_t}{\rho_t}}$.

Using the Taylor expansion of $\ked{\rho_{t+dt}}$ and taking its inner product with $\ked{\rho_t}$, we obtain 
\begin{equation}
    \bradked{\rho_t}{\rho_{t+dt}}=\bradked{\rho_t}{\rho_t}+\bradked{\rho_t}{\dot{\rho}_t}dt+\bradked{\rho_t}{\ddot{\rho}_t}\frac{dt^2}{2}+ {\text{higher order terms}}. 
\end{equation}

Using above expression, we can expand the quantity  $\bradked{\rho_t}{\rho_{t+dt}}\bradked{\rho_{t+dt}}{\rho_t}$ up to second order in $dt$ as 
\begin{equation}
    \bradked{\rho_t}{\rho_{t+dt}}\bradked{\rho_{t+dt}}{\rho_t}=\bradked{\rho_t}{\rho_t}\left(\bradked{\rho_t}{\rho_t}+\left(\bradked{\dot{\rho}_t}{\rho_t}+\bradked{\rho_t}{\dot{\rho}_t}\right)dt+\left(\bradked{\ddot{\rho}_t}{\rho_t}+2\bradked{\rho_t}{\dot{\rho}_t}\bradked{\dot{\rho}_t}{\rho_t}+\bradked{\rho_t}{\ddot{\rho}_t}\right)\frac{dt^2}{2}\right).
\end{equation}
Similarly,
\begin{equation}
    \bradked{\rho_{t+dt}}{\rho_{t+dt}}=\bradked{\rho_t}{\rho_t}+\left(\bradked{\rho_t}{\dot{\rho}_t}+\bradked{\dot{\rho}_t}{\rho_t}\right)dt+\left(\bradked{\rho_t}{\ddot{\rho}_t}+2\bradked{\dot{\rho}_t}{\dot{\rho}_t}+\bradked{\ddot{\rho}_t}{\rho_t}\right)\frac{dt^2}{2}.
\end{equation}
Using the above two equations and considering the terms up to $dt^2$, we obtain the expression for $dS^2$ as follows
\begin{align}\label{genmetric}
    dS^2 &=\left(\frac{\bradked{\dot{\rho}_t}{\dot{\rho}_t}}{\bradked{\rho_t}{\rho_t}}-\frac{\bradked{\rho_t}{\dot{\rho}_t}\bradked{\dot{\rho}_t}{\rho_t}}{\bradked{\rho_t}{\rho_t}^2}\right)dt^2\nonumber\\
    &= \left(\bradked{\dot{\tilde{\rho}}_t}{\dot{\tilde{\rho}}_t}-\bradked{\tilde\rho_t}{\dot{\tilde{\rho}}_t}\bradked{\dot{\tilde{\rho}}_t}{\tilde\rho_t}\right)dt^2,
\end{align}
where $\ked{\dot{\tilde{\rho}}_t}=\frac{d}{dt}\ked{\tilde{\rho}_t}$ and the second equality follows from the fact that  
\begin{equation}
    \frac{\bradked{\dot{\rho}_t}{\dot{\rho}_t}}{\bradked{\rho_t}{\rho_t}}=\bradked{\dot{\tilde{\rho}}_{t}}{\dot{\tilde{\rho}}_{t}}+\frac{(\bradked{\dot{\rho}_t}{\rho_t}+\bradked{\rho_t}{\dot{\rho}_t})^2}{4\bradked{\rho_t}{\rho_t}^2} ~~\text{and} ~~ \frac{\bradked{\rho_t}{\dot{\rho}_t}\bradked{\dot{\rho}_t}{\rho_t}}{\bradked{\rho_t}{\rho_t}^2}=\bradked{\tilde\rho_t}{\dot{\tilde{\rho}}_t}\bradked{\dot{\tilde{\rho}}_t}{\tilde\rho_t}+\frac{(\bradked{\dot{\rho}_t}{\rho_t}+\bradked{\rho_t}{\dot{\rho}_t})^2}{4\bradked{\rho_t}{\rho_t}^2}.
\end{equation}

The Eq.~\eqref{genmetric} can be thought of as Fubini-study metric in Liouville space, where $dS$ measures the elemental distance along the evolution path (trajectory). Using Eq.~\eqref{genmetric}, the speed of the evolution (transportation) of the quantum system in Liouville space can be defined as
\begin{equation}
    \mathcal{V}_t=\frac{dS}{dt}=\sqrt{\bradked{\dot{\tilde{\rho}}_t}{\dot{\tilde{\rho}}_t}-\bradked{\tilde\rho_t}{\dot{\tilde{\rho}}_t}\bradked{\dot{\tilde{\rho}}_t}{\tilde\rho_t}}.
\end{equation}
On integrating the above equation, we obtain
\begin{equation}
    T=\frac{S}{{\langle\!\langle\mathcal{V}_t\rangle\!\rangle_{T}}},~~~\text{where}~~{\langle\!\langle\mathcal{V}_t\rangle\!\rangle_{T}}=\frac{1}{T}\int_{0}^{T} dt~\mathcal{V}_t.
\end{equation}
Now, using the geometry of quantum evolution in Liouville space, the total distance $S$ traveled by the state vector $\ked{\tilde{\rho}_t}$, along the evolution path which joins the state vectors $\ked{\tilde{\rho}_0}$ and $\ked{\tilde{\rho}_T}$ is always lower bounded by the shortest distance (geodesic) connecting them, i.e., $S \geq \Theta(\rho_0,\rho_T)$. Thus, we obtain a lower bound on the total evolution time of a quantum system
\begin{equation} \label{gqslforarbdyn}
T \geq \frac{ \Theta(\rho_0,\rho_{T})}{{\langle\!\langle\mathcal{V}_t\rangle\!\rangle_{T}}}.
\end{equation}

The above bound is applicable to arbitrary time-continuous dynamics. It is important to note that previously, a geometric quantum speed limit has been obtained for arbitrary time-continuous dynamics, which requires purification of the states (see Ref.~\cite{Thakuria2024}). This approach may have some shortcomings because purifications can overestimate the distance between states, and also the speed of evolution. However, the bound in Eq.~\ref{gqslforarbdyn} does not have such shortcomings; therefore, it is relatively better. Moreover, the bound in Eq.~\ref{gqslforarbdyn} is significantly simpler to calculate than other geometric speed limits applicable to arbitrary time-continuous dynamics (see Refs.~\cite{Taddei2013, Pires2016}).

{\it \bf Geometrical derivation Inexact Quantum Speed Limit for CPTP dynamics (Alternate Proof of Theorem-1):} For the CPTP dynamics given by Eq.~\eqref{Master}, we can write the speed of evolution as 
\begin{equation}\label{CS}
    \mathcal{V}_t=\Delta\mathcal{L}=\sqrt{\tr\left(\mathcal{L}^\dagger\mathcal{L}\mathcal{P}_t\right)-\tr\left(\mathcal{L}\mathcal{P}_t\right)\tr\left(\mathcal{L}^\dagger\mathcal{P}_t\right)}.
\end{equation}
To arrive at Eq.~\eqref{CS} from Eq.~\eqref{gqslforarbdyn}, we have used $\bradked{\dot{\tilde{\rho}}_t}{\dot{\tilde{\rho}}_t}=\tr(\mathcal{L}^\dagger\mathcal{L}\mathcal{P}_t)-(\tr\left(\mathcal{L}\mathcal{P}_t\right)+\tr(\mathcal{L}^\dagger\mathcal{P}_t))^2/4$ and $\bradked{\tilde\rho_t}{\dot{\tilde{\rho}}_t}\bradked{\dot{\tilde{\rho}}_t}{\tilde\rho_t}=\tr\left(\mathcal{L}\mathcal{P}_t\right)\tr(\mathcal{L}^\dagger\mathcal{P}_t)-(\tr\left(\mathcal{L}\mathcal{P}_t\right)+\tr(\mathcal{L}^\dagger\mathcal{P}_t))^2/4$. Using this evolution speed in Eq.~\eqref{gqslforarbdyn}, we obtain
\begin{equation}
    T\geq \frac{\Theta(\rho_0,\rho_{T})}{\langle\!\langle\Delta\mathcal{L}\rangle\!\rangle_{T}},
\end{equation}
where  $\langle\!\langle\Delta\mathcal{L}\rangle\!\rangle_{T}= \frac{1}{T}\int_0^T dt \Delta \mathcal{L}$.\\

{\it \bf Geometrical derivation of Inexact Quantum Speed Limit for CPTP map described by operator-sum representation:}
The CPTP dynamics of a quantum system can also be described by the operator-sum representation, which is given as:
\begin{equation}
    \rho_t=\sum_i K_i(t) \rho_0 K^\dagger_i(t),
\end{equation}
where $\{K_{i}(t)\}$ are called Kraus operators and satisfy the completeness relation $\sum_{i} K_{i}^{\dag}(t) K_{i}(t) = \mathbb{I}$. The vectorized form of the above equation can be written as $\ked{\rho_t}=\mathcal{K}_t\ked{\rho_0}$, where $\mathcal{K}_{t}=\sum_i K_i(t) \otimes K^{*}_i(t)$ is the corresponding operator that acts on $\ked{\rho_0}$ in Liouville space. The normalized state evolves as $\ked{\tilde{\rho}_t} =\widetilde{\mathcal{K}}_t\ked{\tilde{\rho}_0}$, where $\widetilde{\mathcal{K}}_t=\mathcal{K}_t/\sqrt{\bradked{\rho_0 \mathcal{K}^\dagger_t}{\mathcal{K}_t\rho_0}}$. Using this equation we obtain $\bradked{\dot{\tilde{\rho}}_t}{\dot{\tilde{\rho}}_t}=\tr(\dot{\widetilde{\mathcal{K}}}^\dagger_t\dot{\widetilde{\mathcal{K}}}_t\mathcal{P}_0)$ and $\bradked{\tilde\rho_t}{\dot{\tilde{\rho}}_t}\bradked{\dot{\tilde{\rho}}_t}{\tilde\rho_t}=\tr(\dot{\widetilde{\mathcal{K}}}^\dagger_t\mathcal{P}_t\dot{\widetilde{\mathcal{K}}}_t\mathcal{P}_0)$. Hence, the evolution speed of the system in Liouville space in terms of the Kraus operators can be written as 
\begin{equation}
\mathcal{V}_t=\sqrt{\tr(\dot{\widetilde{\mathcal{K}}}^\dagger_t\dot{\widetilde{\mathcal{K}}}_t\mathcal{P}_0)-\tr(\dot{\widetilde{\mathcal{K}}}^\dagger_t\mathcal{P}_t\dot{\widetilde{\mathcal{K}}}_t\mathcal{P}_0)}.
\end{equation}
Using the above equation, we obtain the lower bound on the evolution time as
\begin{equation}
    T\geq \frac{\Theta(\rho_0,\rho_{T})}{\left\langle\!\left\langle\sqrt{\tr(\dot{\widetilde{\mathcal{K}}}^\dagger_t\dot{\widetilde{\mathcal{K}}}_t\mathcal{P}_0)-\tr(\dot{\widetilde{\mathcal{K}}}^\dagger_t \mathcal{P}_t\dot{\widetilde{\mathcal{K}}}_t\mathcal{P}_0)} \right\rangle\!\right\rangle_{T}},
\end{equation}

where \(\langle\!\langle X_t \rangle\!\rangle_{T} := \frac{1}{T} \int_{0}^{T} \, dt X_t \).

\section{Proof of Corollary}\label{appendixD}
To prove the saturation of the generalized Mandelstam-Tamm bound~\ref{mainbound} (Eq. (4) in the main text), we first show when Eq.~\eqref{equ:rateofdis} saturates. Let us recall Eq.~\eqref{equ:rateofdis} 
 \begin{eqnarray}\label{boundonrateofdis}
    \left|\frac{d}{dt}\tr(\mathcal{P}_t \mathcal{P}_0)\right| 
    \leq2|\tr(\mathcal{L}\mathcal{P}_t \mathcal{P}_0)-\tr(\mathcal{P}_t\mathcal{P}_0)\tr(\mathcal{L}\mathcal{P}_t)|
    \leq2\Delta\mathcal{L}\Delta\mathcal{P}_0.
\end{eqnarray}
Let us define two vectors in Liouville space as $\ked{g} := \left(\mathcal{L}-\tr{(\mathcal{L}\mathcal{P}_{t})}\right)\ked{\tilde{\rho}_t}$ and $\ked{f} := \left(\mathcal{P}_{0}-\tr{(\mathcal{P}_{0}\mathcal{P}_{t})}\right)\ked{\tilde{\rho}_t}$, then the above inequality can be rewritten as 
\begin{eqnarray}
    \left|\frac{d}{dt}\tr(\mathcal{P}_t \mathcal{P}_0)\right| 
    \leq2|\tr(\mathcal{L}\mathcal{P}_t \mathcal{P}_0)-\tr(\mathcal{P}_t\mathcal{P}_0)\tr(\mathcal{L}\mathcal{P}_t)|=2|\bradked{f}{g}|
    \leq2\sqrt{\bradked{g}{g}}\sqrt{\bradked{f}{f}}.
\end{eqnarray}
Here, the first inequality is due to the triangle inequality of the form $\abs{\bradked{f}{g}+\bradked{g}{f}}\leq2\abs{\bradked{f}{g}}$ (see Eq.~\eqref{triangine}), which saturates if $\ked{g} = a \ked{f}$, for some scalar $a\in\mathbb{R}$. The second inequality is due to the Cauchy-Schwarz inequality, which saturates if $\ked{g} = a \ked{f}$, for some $a\in \mathbb{C}$. Thus, to saturate both of these inequalities $a\in\mathbb{R}$. Hence, we obtain
 \begin{equation}
     \mathcal{L}\ked{\tilde{\rho}_{t}} - \brad{\tilde{\rho}_{t}}\mathcal{L}\ked{\tilde{\rho}_{t}}\ked{\tilde{\rho}_{t}} = a \bradked{\tilde{\rho}_{0}}{\tilde{\rho}_{t}}\left(\ked{\tilde{\rho}_{0}} - \bradked{\tilde{\rho}_{0}}{\tilde{\rho}_{t}} \ked{\tilde{\rho}_{t}}\right),~~~~~~~a\in\mathbb{R}.\label{equ_a}
 \end{equation}
The above equation provides a condition on the saturation of Eq.~\eqref{boundonrateofdis} in terms of a vector equation. By taking the inner product of the above equation with its complex conjugate, we obtain 
\begin{eqnarray}\label{valueofa}
    a=\pm\frac{\Delta \mathcal{L}}{\bradked{\tilde{\rho}_0}{\tilde{\rho}_t}\sqrt{1-\bradked{\tilde{\rho}_0}{\tilde{\rho}_t}^2}}.
\end{eqnarray}
Of the two values of $a$ given in the above equation, only one satisfies Eq.~\eqref{equ_a}, which will be determined by the dynamics of the system. This fixes the value of $a$ at any time $t$ for a given Liouvillian $\mathcal{L}$ and initial state $\ked{\tilde{\rho}_0}$. Now, using Eq.~\eqref{equ_a} along with Eq.~\eqref{eomfornormrho}, we obtain the saturation condition for the bound in Eq.~ \ref{boundonrateofdis}, which is given by the following dynamical equation
\begin{equation} \label{satcondondyn}
    \ked{\dot{\tilde{\rho}}_t}=\rm{sgn}(a)\frac{\Delta \mathcal{L}}{\sqrt{1-\bradked{\tilde{\rho}_0}{\tilde{\rho}_t}^2}}\left(\ked{\tilde{\rho}_{0}} - \bradked{\tilde{\rho}_{0}}{\tilde{\rho}_{t}}\ked{{\tilde\rho_t}}\right),
\end{equation}
where $\sgn(a)$ denotes the sign of $a$. Now, by the Aharonov-Vaidman identity~\cite{Aharonov1990}, we have
\begin{equation}
     \mathcal{L}\ked{\tilde{\rho}_t}=\bradked{\tilde{\rho}_t}{\mathcal{L}|\tilde{\rho}_t}\ked{\tilde{\rho}_t}+ \Delta \mathcal{L}\ked{\tilde{\rho}^{\perp}_{t}},
\end{equation}
where $\ked{\tilde{\rho}^{\perp}_{t}}$ depends on $\mathcal{L}$, and $\bradked{\tilde{\rho}_t}{\tilde{\rho}^{\perp}_{t}}=0$. We can rewrite the above identity using Eq.~\eqref{eomfornormrho} as follows
\begin{equation}\label{avidentity}
    \ked{\dot{\tilde{\rho}}_t}=\Delta \mathcal{L}\ked{\tilde{\rho}^{\perp}_{t}}.
\end{equation}
Using the above relation along with the saturation condition in Eq.~\ref{satcondondyn}, we obtain
 \begin{equation}\label{rhotperpdef}
    \ked{\tilde{\rho}^{\perp}_{t}}= \sgn(a)\frac{1}{\sqrt{1-\alpha_t^2}}\left(\ked{\tilde{\rho}_{0}} - \alpha_t \ked{\tilde{\rho}_{t}}\right),
 \end{equation}
where we have defined $\alpha_t:=\bradked{\tilde{\rho}_{0}}{\tilde{\rho}_{t}}$ for notational simplicity. Let us now take the time derivative of the above equation
\begin{eqnarray}\label{derivativeofrhotperp}
    \ked{\dot{\tilde{\rho}}^{\perp}_{t}}&=&\sgn(a)\left[\frac{\alpha_{t}\dot{\alpha}_{t}}{(1-\alpha_t^2)^{3/2}} \left(\ked{\tilde{\rho}_{0}} - \alpha_t \ked{\tilde{\rho}_{t}}\right)+\frac{1}{\sqrt{1-\alpha_t^2}}\left(-\alpha_t \ked{\dot{\tilde{\rho}}_t} - \dot{\alpha}_t \ked{\tilde{\rho}_t} \right)\right]\nonumber\\
    &=&\sgn(a)\left[\frac{\alpha_{t}\dot{\alpha}_{t}}{1-\alpha_t^2} \ked{\tilde{\rho}^{\perp}_{t}}+\frac{1}{\sqrt{1-\alpha_t^2}}\left(-\alpha_t \Delta \mathcal{L}\ked{\tilde{\rho}^{\perp}_{t}} - \dot{\alpha}_t \ked{\tilde{\rho}_t} \right)\right],
\end{eqnarray}
where in the second equality, we have used Eq.~\eqref{avidentity}. Now by rearranging Eq.~\eqref{rhotperpdef}, we obtain
\begin{equation}\label{rhoinitialsat}
    \ked{\tilde{\rho}_0}=\alpha_t\ked{\tilde{\rho}_{t}}+\sgn(a)\sqrt{1-\alpha_t^2}\ked{\tilde{\rho}^{\perp}_{t}}.
\end{equation}
Using the above equation, we can write a normalized orthogonal vector $\ked{\tilde{\rho}_0^{\perp}}$ as follows
\begin{equation}\label{rhoinitialperpsat}
    \ked{\tilde{\rho}^{\perp}_{0}}=\sgn(a)\sqrt{1-\alpha_t^2} \ked{\tilde{\rho}_{t}}-\alpha_t\ked{\tilde{\rho}^{\perp}_{t}},
\end{equation}
where $\bradked{\tilde{\rho}_0}{\tilde{\rho}^{\perp}_{0}}=0$. Next, we show that $\ked{\tilde{\rho}^{\perp}_{0}}$ is necessarily time-independent. On taking the derivative of the above equation with respect to time $t$, we obtain
\begin{align}
    \ked{\dot{\tilde{\rho}}^{\perp}_{0}} &=-\sgn(a)\frac{\alpha_{t}\dot{\alpha}_{t}}{\sqrt{1-\alpha_t^2}} \ked{\tilde{\rho}_t}+\sgn(a)\sqrt{1-\alpha_t^2} \ked{\dot{\tilde{\rho}}_t} - \dot{\alpha}_{t} \ked{\tilde{\rho}^{\perp}_{t}}-\alpha_{t}\ked{\dot{\tilde{\rho}}^{\perp}_{t}}\nonumber\\
    &= -\sgn(a)\frac{\alpha_{t}\dot{\alpha}_{t}}{\sqrt{1-\alpha_t^2}} \ked{\tilde{\rho}_t}+\sgn(a)\sqrt{1-\alpha_t^2} \Delta \mathcal{L} \ked{\tilde{\rho}^{\perp}_{t}}-\dot{\alpha}_t \ked{\tilde{\rho}^{\perp}_{t}} - \alpha_t \sgn(a)\left(\frac{\alpha_{t}\dot{\alpha}_{t}}{1-\alpha_t^2} \ked{\tilde{\rho}^{\perp}_{t}}+\frac{1}{\sqrt{1-\alpha_t^2}}\left(-\alpha_t \Delta \mathcal{L}\ked{\tilde{\rho}^{\perp}_{t}} - \dot{\alpha}_t \ked{\tilde{\rho}_t} \right)\right)\nonumber\\
    & = \sgn(a)\left(\frac{\alpha_{t}\dot{\alpha}_{t}}{\sqrt{1- \alpha_t^2}}  - \frac{\alpha_{t}\dot{\alpha}_{t}}{\sqrt{1-\alpha_t^2}} \right)\ked{\tilde{\rho}_t} + \left(\sgn(a)\sqrt{1-\alpha_t^2} \Delta \mathcal{L} -\dot{\alpha}_t - \sgn(a)\frac{\alpha_t^2\dot{\alpha}_t}{1-\alpha_t^2} + \frac{\alpha_t^2\Delta \mathcal{L}}{\sqrt{1- \alpha_t^2}}  \right) \ked{\tilde{\rho}^{\perp}_{t}}\nonumber\\
    &=0,
\end{align}
where in the second equality, we have used Eqs.~\eqref{avidentity} and \eqref{derivativeofrhotperp}, and in the third equality we rearranged the coefficients of $\ked{\tilde{\rho}_t}$ and $\ked{\tilde{\rho}^{\perp}_{t}}$. Finally, to obtain the last equality, we have used the fact that $\dot{\alpha}_t=\bradked{\tilde{\rho}_{0}}{\dot{\tilde{\rho}}_{t}}=\bradked{\tilde{\rho}_{0}}{\tilde{\rho}^{\perp}_{t}}\Delta \mathcal{L}=\sgn(a)\sqrt{1-\alpha_t^2}\Delta\mathcal{L}$. Therefore, $\ked{\tilde{\rho}^{\perp}_{0}}$ does not change in time. Now, using Eqs.~\eqref{rhoinitialsat} and \eqref{rhoinitialperpsat}, we can write the time-evolved normalized state as
\begin{equation}
\ked{\tilde{\rho}_t}=\bradked{\tilde{\rho}_{0}}{\tilde{\rho}_{t}} \ked{\tilde{\rho}_{0}}+\sgn(a)\sqrt{1-\bradked{\tilde{\rho}_0}{\tilde{\rho}_t}^2}\ked{\tilde{\rho}^{\perp}_{0}}.
\end{equation}
The above equation says that for the dynamics that saturate the bound in Eq.~\ref{boundonrateofdis}, the evolution of the state is such that $\ked{\tilde{\rho}_t}$ always resides in the two-dimensional subspace spanned by $\{\ked{\tilde{\rho}_{0}}, \ked{\tilde{\rho}^{\perp}_{0}}\}$. 
Note that the vector $\ked{\tilde{\rho}^{\perp}_{0}}$ may not correspond to a physical state in general. In that case, the state $\ked{\tilde{\rho}_t}$ of the system will never reach the vector $\ked{\tilde{\rho}^{\perp}_{0}}$. To see the conditions under which $\ked{\tilde{\rho}^{\perp}_{0}}$ becomes a physical state, let us write the above equation in the space of linear operators as
\begin{equation}\label{rho0perpmatrixform}
\frac{\rho_t}{\sqrt{\tr(\rho_t^2)}}=\frac{\tr(\rho_0\rho_t)}{\sqrt{\tr(\rho_0^2)}\sqrt{\tr(\rho_t^2)}}\frac{\rho_0}{\sqrt{\tr(\rho_0^2)}}+\sgn(a)\frac{\sqrt{\tr(\rho_0^2)\tr(\rho_t^2)-\tr(\rho_0\rho_t)^2}}{\sqrt{\tr(\rho_0^2)}\sqrt{\tr(\rho_t^2)}}\frac{\rho^{\perp}_{0}}{\sqrt{\tr(\rho^{\perp}_0)^2}}.
\end{equation}
It is clear from the above equation that $\rho^{\perp}_0$ is Hermitian. Moreover, by taking trace of the above equation, we obtain
\begin{equation}
    \frac{\tr(\rho^{\perp}_0)}{\sqrt{\tr(\rho^{\perp}_0)^2}}=\sgn(a)\frac{\tr(\rho_0^2)-\tr(\rho_0\rho_t)}{\sqrt{\tr(\rho_0^2)}\sqrt{\tr(\rho_0^2)\tr(\rho_t^2)-\tr(\rho_0\rho_t)^2}}.
\end{equation}
Thus, $\tr(\rho^{\perp}_0)=1$ if
\begin{equation}\label{cond1onrhoperpstate}
\sqrt{\tr(\rho^{\perp}_0)^2}=\sgn(a)\frac{\sqrt{\tr(\rho_0^2)}\sqrt{\tr(\rho_0^2)\tr(\rho_t^2)-\tr(\rho_0\rho_t)^2}}{\tr(\rho_0^2)-\tr(\rho_0\rho_t)}.
\end{equation}
If the dynamics of the system is such that Eqs.~\eqref{cond1onrhoperpstate} is satisfied then Eq.~\eqref{rho0perpmatrixform} can be rewritten as
\begin{equation}\label{geodesic}
    \rho_t=\frac{\tr(\rho_0\rho_t)}{\tr(\rho_0^2)}\rho_0+\left(1-\frac{\tr(\rho_0\rho_t)}{\tr(\rho_0^2)}\right)\rho^{\perp}_0,
\end{equation}
which traces a line in Liouville space. Hence, the dynamics keep the time-evolved state $\rho_t$ on a line between the initial state $\rho_0$ and its orthogonal operator $\rho^{\perp}_0$. Moreover, the operator $\rho_0^\perp$ corresponds to a physical state if it is also positive-semidefinite, i.e., if
\begin{eqnarray} \label{cond2onrhoperpstate}
    \braket{\phi}{\rho_t|\phi}-\frac{\tr(\rho_0\rho_t)}{\tr(\rho_0^2)}\braket{\phi}{\rho_0|\phi}\geq0~~~~~~~~\forall \phi.
\end{eqnarray}
The evolution given by Eq.~\eqref{geodesic} saturates the
the bound in Eq.~\ref{boundonrateofdis}. However, to saturate the generalized Mandelstam-Tamm bound~\ref{mainbound} (Eq. (4) in the main text), we require $P_t:=\frac{\tr(\rho_0\rho_t)}{\tr(\rho_0^2)}$ to be monotonically decreasing as we have used an additional inequality $\left(\abs{\int f(t)dt}\leq\int \abs{f(t)}dt\right)$ to arrive at the generalized Mandelstam-Tamm bound from Eq.~\eqref{equ:rateofdis}. It is important to note that Eq.~\eqref{geodesic} connects the given initial and final states ($\rho_0$ and $\rho^\perp_0$) via geodesic in Liouville space. We refer the underlying CPTP dynamics as optimal CPTP dynamics as it transports the state of the system through the minimal path, i.e., geodesic in Liouville space. In the following, we present the form of the Liouvillian that generates the optimal CPTP dynamics for given initial and final states. \\

{\it \bf Liouvillian and Kraus operators for optimal CPTP dynamics when initial and final states are orthogonal:} Let the initial and final states of a $d$-dimensional quantum system be $\rho_0$ and  $\rho^\perp_0$, respectively, where $\tr(\rho_0\rho^\perp_0)=0$. Then using Eq.~\eqref{geodesic}, we have
\begin{equation}\label{channelforgeodesic}
\rho_t=\mathcal{E}_t(\rho_0)=\sum_iK_i(t)\rho_0K_i^\dagger(t)={P}_t\rho_0+(1-{P}_t)\rho^\perp_0,
\end{equation}
where $\{K_i(t)\}$ are the Kraus operators of the CPTP map $\mathcal{E}_t$ and ${P}_t=\frac{\tr(\rho_0\rho_t)}{\tr(\rho_0^2)}$. It is easy to see that the Kraus operators of $\mathcal{E}_t$ are $K_0(t)=\sqrt{P_t}~\mathbb{I}$ and $K_1(t)=\sqrt{1-{P}_t}U$, where $U$ is unitary such that $\rho^\perp_0=U\rho_0U^\dagger$. We know that the Kraus operators of any CPTP dynamical map can be written in terms of the Hamiltonian ($H$) and the Lindblad operators (\{$L_k$\}) of the corresponding master equation under the Born-Markov approximation as
\begin{eqnarray}\label{kraus-liouvillan}
    K_0(t)&=&\mathbb{I}+\left(-i H-\frac{1}{2}\sum_{k\geq1}\gamma_k L_k^\dagger L_k\right)dt,\nonumber\\
    K_k(t)&=&\sqrt{\gamma_k}L_k\sqrt{dt},~~~~k\geq1.
    \end{eqnarray}
Using the above equation and the Kraus operators of 
$\mathcal{E}_t$, we obtain
\begin{eqnarray}
    H&=&0,\nonumber\\
    L_1&=&U,~~~~L_{k\geq2}=0,
\end{eqnarray}
and $P_t=1-\gamma dt +\mathcal{O}(dt^2)$, where $\gamma=\gamma_1$. 
We can now write the Liouvillian using Eq.~\eqref{decompofL} as
\begin{equation}
    \mathcal{L}=\gamma\left(U^{*}\otimes U-\mathbb{I}\otimes\mathbb{I}\right),
\end{equation}
where we have used the fact $L_1^\dagger L_1=\mathbb{I}$. This Liouvillian $\mathcal{L}$ generates a CPTP dynamics that saturates the generalized Mandelstam-Tamm bound \ref{mainbound} (Eq. (4) in the main text) for given initial and final states.\\

{\it \bf Liouvillian and Kraus operators for optimal CPTP dynamics when initial and final states are pure and orthogonal:} If the initial and final states of a $d$-dimensional quantum system are $\rho_0=\ketbra{\psi}$ and $\rho^{\perp}_0=\ketbra{\psi^{\perp}}$, where $\braket{\psi}{\psi^{\perp}}=0$, then the Kraus operators of the CPTP map can be written as $\{K_0(t)=\sqrt{P_t}~\mathbb{I}~~ \text{and}~~K_1(t)=\sqrt{1-{P}_t}\left(\ketbra{\psi}{\psi^\perp}+\ketbra{\psi^\perp}{\psi}\right)$. Therefore, the Hamiltonian and the Lindblad operators are
\begin{eqnarray}
    H&=&0,\nonumber\\
    L_1&=&\left(\ketbra{\psi}{\psi^\perp}+\ketbra{\psi^\perp}{\psi}\right),~~~~L_{k\geq2}=0.
\end{eqnarray}
We can now write the Liouvillian using Eq.~\eqref{decompofL} as
\begin{equation}
    \mathcal{L}=\gamma\left(\left(\ketbra{\psi}{\psi^\perp}+\ketbra{\psi^\perp}{\psi}\right)^{*}\otimes \left(\ketbra{\psi}{\psi^\perp}+\ketbra{\psi^\perp}{\psi}\right)-\mathbb{I}\otimes\mathbb{I}\right).
\end{equation}

\section{Proof of Theorem~2}\label{appendixE}

Let us consider a $d$-dimensional quantum system in the initial state $ \rho_0$ evolving under a CPTP dynamics generated by Liouvillian $\mathcal{L}$. Using the formalism discussed in Appendix~\ref{appendixA}, we can decompose the Liouvillian as $\mathcal{L}=\mathcal{L}_{\rm cl}+\mathcal{L}_{\rm nc}$, where $\mathcal{L}_{\rm cl}$ is diagonal in the eigenbasis of a Hermitian operator $\mathcal{A}=\sum_i a_i\kedbrad{a_i}{a_i}$ and is defined as
 \begin{equation}
  \mathcal{L}_{\rm cl}:=\sum_i \kedbrad{a_i}{a_i}\frac{(a_i|\frac{1}{2}(\mathcal{L}\mathcal{P}_t-\mathcal{P}_t\mathcal{L}^\dagger)|a_i)}{(a_i|\mathcal{P}_t|a_i)},
\end{equation}
where $\mathcal{P}_t=\kedbrad{\tilde{\rho}_t}{\tilde{\rho}_t}$ and $\ked{\tilde{\rho}_{t}}$ is time-evolved state vector. The decomposition of $\mathcal{L}=\mathcal{L}_{\rm cl}+\mathcal{L}_{\rm nc}$ is such that $\mathcal{L}_{\rm cl}$ is uncorrelated with $\mathcal{L}_{\rm nc}$ in any state $\ked{\tilde{\rho}_t}$, i.e., $\Delta \mathcal{L}^{2}=\Delta \mathcal{L}_{\rm cl}^{2}+\Delta \mathcal{L}_{\rm nc}^2$. In the following, we take $\mathcal{A}=\mathcal{P}_0(=\kedbrad{\tilde{\rho}_0}{\tilde{\rho}_0})$, which implies that $\mathcal{L}_{\rm cl}$ commutes with $\mathcal{P}_0$.

Using the exact uncertainty relation for operators acting on Liouville space, $\mathcal{A}(=\mathcal{P}_0)$ and $\mathcal{B}(=\mathcal{L})$ given in Eq.~\eqref{exactuncrln1}, we have
\begin{equation}
     (\delta_{\mathcal{L}}\mathcal{P}_0)^{-2}=\sum_{i=0}^{d^2-1}\frac{ \brad{a_i}\frac{d}{dt}\mathcal{P}_t\ked{a_i}^2}{(a_i|\mathcal{P}_t|a_i)} = 4  \Delta \mathcal{L}_{\rm nc}^2 ,
\end{equation}
where $\frac{d}{dt}\mathcal{P}_t=(\mathcal{L}\mathcal{P}_t+\mathcal{P}_t \mathcal{L}^\dagger-\mathcal{P}_t\tr[\mathcal{P}_t(\mathcal{L}+\mathcal{L}^\dagger)])$ and $\{\ked{a_i}\}$ form an orthonormal basis containing $\ked{\tilde{\rho}_0}$ as one of the basis elements. Since $\{\ked{a_i}\}$ are time independent, we have $ \bradked{a_i}{d\mathcal{P}_t|a_i}=d\bradked{a_i}{\mathcal{P}_t|a_i}=d\bradked{a_i}{\tilde{\rho_t}}\bradked{\tilde{\rho_t}}{a_i}=d\abs{c_i}^2$ where $c_i=\bradked{a_i}{\tilde{\rho}_t}$. Hence, we have
\begin{equation}
    \Delta \mathcal{L}_{\rm nc}^2 dt^2=\frac{1}{4}\sum_{i=0}^{d^2-1} \frac{(d\abs{c_i}^2)^2}{\abs{c_i}^2}=\sum_{i=0}^{d^2-1} (d\abs{c_i})^2.
\end{equation}
Integrating the above equation with respect to time, we obtain
\begin{equation}
    T=\frac{1}{\langle\!\langle\Delta \mathcal{L}_{\rm nc}\rangle\!\rangle}\int_0^T \sqrt{\sum_{i=0}^{d^2-1} \left(\frac{d\abs{c_i}}{dt}\right)^2}dt,
\end{equation}
where $\langle\!\langle\Delta \mathcal{L}_{\rm nc}\rangle\!\rangle=\frac{1}{T}\int_0^T  dt \Delta \mathcal{L}_{\rm nc}$ is the time-averaged $\Delta\mathcal{L}_{\rm nc}$. Let us define a real vector $\ked{\sigma_t}=\sum_{i=0}^{d^2-1} \abs{c_i}\ked{a_i}$ in Liouville space, then $\bradked{\dot{\sigma_t}}{\dot{\sigma_t}}=\sum_{i=0}^{d^2-1} (\frac{d}{dt}\abs{c_i})^2$. The length of the path traced out by the vector 
$\ked{\sigma_t}$ in Liouville space during the evolution is given by
\begin{equation}
    l(\rho_0,\rho_{T})=\int_0^T \sqrt{\bradked{\dot{\sigma_t}}{\dot{\sigma_t}}}dt=\int_0^T\sqrt{\sum_{i=0}^{d^2-1} \left(\frac{d\abs{c_i}}{dt}\right)^2}~dt. \label{equ:Llength}
\end{equation}
It is clear that $ l(\rho_0,\rho_{T})$ is the path length between $\ked{\sigma_0} = \ked{\rho_0}$ and $\ked{\sigma_T} = \sum_{i=0}^{d^2-1} \abs{\bradked{a_i}{\tilde{\rho}_T}}\ked{a_i}$. Using the above two equations, we get the exact time of evolution as
\begin{equation}
  T=\frac{ l(\rho_0,\rho_{T})}{\langle\!\langle\Delta \mathcal{L}_{\rm nc}\rangle\!\rangle}.  
\end{equation}
If the evolution is confined in the two-dimensional Liouville subspace spanned by  $\{\ked{a_0}=\ked{\tilde{\rho}_{0}}, \ked{a_1} = \ked{\tilde{\rho}^{\perp}_{0}}\}$ and $c_i$ are monotonic functions, we have $ l(\rho_0,\rho_{T})=\Theta(\rho_0,\rho_{T})$, otherwise $l(\rho_0,\rho_{T})\geq\Theta(\rho_0,\rho_{T})$. Thus, we obtain the following inequality 
\begin{equation}
  T=\frac{ l(\rho_0,\rho_{T})}{\langle\!\langle\Delta \mathcal{L}_{\rm nc}\rangle\!\rangle} \geq \frac{ \Theta(\rho_0,\rho_{T})}{\langle\!\langle\Delta\mathcal{L}_{\rm nc}\rangle\!\rangle_{T}}.
\end{equation}
Let us recall the Wootters distance \(S_{\rm PD}\) in the space of probability distributions, defined as:
\[
S_{\rm PD} = \int_{\gamma} dS_{\rm PD},
\]
where 
\[
(dS_{\text{PD}})^2 = \frac{1}{4} \sum_i \frac{(dp_i)^2}{p_i}
= \frac{1}{4} F_{cl}dt^2.\]
Here, \(p_i\) are probabilities, $F_{cl}$ is classical Fisher information, and \(\gamma\) is the curve traced out by a quantum system in the space of probability distributions when observed in some basis~\cite{Wootters1981, Braunstein1994, Miller2023}. Since \(|c_i|^2\) in Eq.~\eqref{equ:Llength} can be thought of as associated probabilities corresponding to the basis \(\{\ked{a_i}\}\) and the quantity $\bradked{\dot{\sigma_t}}{\dot{\sigma_t}}$ can be referred to as the associated classical Fisher information $\mathcal{F}_{cl}$ in Liouville space. Hence, \( l(\rho_0, \rho_T)\) can be regarded as analogous to the Wootters distance in Liouville space.\\

{\it \bf Alternative Proof of the Lower Bound in Theorem~2:} Using the decomposition of the Liouvillian, $\mathcal{L}=\mathcal{L}_{\rm cl}+\mathcal{L}_{\rm nc}$, in Eq.~\eqref{Dist}, we have
\begin{eqnarray}
\frac{d}{dt}\tr(\mathcal{P}_t \mathcal{P}_0)&=&\tr\left((\mathcal{L}_{\rm cl}+\mathcal{L}_{\rm nc})\mathcal{P}_t\mathcal{P}_0\right)+\tr\left(\mathcal{P}_t(\mathcal{L}_{\rm cl}+\mathcal{L}_{\rm nc})^\dagger\mathcal{P}_0\right)-\tr\left(\mathcal{P}_t\mathcal{P}_0\right)\tr\left(\left((\mathcal{L}_{\rm cl}+\mathcal{L}_{\rm nc})+(\mathcal{L}_{\rm cl}+\mathcal{L}_{\rm nc})^\dagger\right)\mathcal{P}_t\right)\nonumber\\
&=&\tr(\mathcal{L}_{\rm cl}\mathcal{P}_t\mathcal{P}_0-\mathcal{P}_t\mathcal{L}_{\rm cl}\mathcal{P}_0)+\tr(\mathcal{L}_{\rm nc}\mathcal{P}_t\mathcal{P}_0+\mathcal{P}_t\mathcal{L}^{\dagger}_{\rm nc}\mathcal{P}_0)-\tr(\mathcal{P}_t\mathcal{P}_0)\tr((\mathcal{L}_{\rm nc}+\mathcal{L}^{\dagger}_{\rm nc})\mathcal{P}_t)\nonumber\\
&=&\tr(\mathcal{P}_t[\mathcal{P}_0,\mathcal{L}_{\rm cl}])+\tr(\mathcal{L}_{\rm nc}\mathcal{P}_t\mathcal{P}_0+\mathcal{P}_t\mathcal{L}^{\dagger}_{\rm nc}\mathcal{P}_0)-\tr(\mathcal{P}_t\mathcal{P}_0)\tr((\mathcal{L}_{\rm nc}+\mathcal{L}^{\dagger}_{\rm nc})\mathcal{P}_t),
\end{eqnarray}
where in the second equality, we have used the fact that $\mathcal{L}^{\dagger}_{\rm cl}=-\mathcal{L}_{\rm cl}$. Now, due to the cyclic property of trace and the fact that $[\mathcal{L}_{\rm cl},\mathcal{P}_0]=0$, the first term of the above equation vanishes. Thus, we obtain
\begin{eqnarray}
  \frac{d}{dt}\tr(\mathcal{P}_t\mathcal{P}_0)&=&\tr(\mathcal{L}_{\rm nc}\mathcal{P}_t\mathcal{P}_0+\mathcal{P}_t\mathcal{L}^{\dagger}_{\rm nc}\mathcal{P}_0)-\tr(\mathcal{P}_t\mathcal{P}_0)\tr((\mathcal{L}_{\rm nc}+\mathcal{L}^{\dagger}_{\rm nc})\mathcal{P}_t).
\end{eqnarray}
The above equation implies that $\mathcal{L}_{\rm cl}$ does not contribute to the rate of change of $\tr(\mathcal{P}_t\mathcal{P}_0)$. Now, if we replace $\mathcal{L}$ by $\mathcal{L}_{\rm cl}$ in Proof of Theorem~1 (see Appendix~\ref{appendixB}), we obtain the following bound
\begin{equation}
T \geq \frac{ \Theta(\rho_0,\rho_{T})}{\langle\!\langle\Delta\mathcal{L}_{\rm nc}\rangle\!\rangle_{T}}.
\end{equation}
The above bound is tighter than the bound \ref{mainbound} (Eq. (4) in the main text) because $\Delta\mathcal{L}_{\rm nc}=\sqrt{(\Delta\mathcal{L})^2-(\Delta\mathcal{L}_{\rm cl})^2}$, which implies that $\Delta\mathcal{L}_{\rm nc}\leq\Delta\mathcal{L}$. Hence, $\Delta\mathcal{L}_{\rm nc}$ can be thought of as the refined speed of evolution. 

\section{Bound on Krylov complexity}\label{appendixF}
The time evolution of a density operator in the Krylov space is given by $\ked{\rho_t}=\sum_n i^n \phi_n(t) \ked{K_n}$, where $\{\ked{K_n}\}$ form the Krylov basis, and $\phi_n(0)=\delta_{n0}$.
The Krylov complexity is defined as $ C_{K}(t) :=\tr(\mathcal{P}_t\mathcal{S})$,
where $\mathcal{P}_t=\kedbrad{\rho_t}{\rho_t}/\tr(\rho_0^2)$ and $\mathcal{S}=\sum_n n \kedbrad{K_n}{K_n}$. Since $\tr\left(\mathcal{P}_0 \mathcal{S}\right) =\sum_n \brad{K_0} \mathcal{S} \ked{K_0} = 0$, we have
\begin{align}
    C_K&=\abs{\tr(\mathcal{P}_t-\mathcal{P}_0)\mathcal{S}}\\
    &  \leq\norm{\mathcal{P}_t-\mathcal{P}_0}_{\rm tr}\norm{\mathcal{S}}_{\rm op}\\
    &=2\sqrt{1-\tr(\mathcal{P}_t\mathcal{P}_0)}\norm{\mathcal{S}}_{\rm op},
\end{align}
where in the second line we have used H$\ddot{o}$lder's inequality~\cite{Holder1889}, which for two operators $\mathcal{A}$ and $\mathcal{B}$ is given as $|\tr(\mathcal{A}^\dagger\mathcal{B})|\leq\norm{\mathcal{A}}_{p} \norm{\mathcal{B}}_{q}$,
where $\norm{\mathcal{A}}_{p}=\left[\tr (\vert \mathcal{A}\vert^{p})\right]^{\frac{1}{p}}$ is the Schatten $p$-norm of an operator $\mathcal{A}$, $\vert\mathcal{A}\vert = \sqrt{\mathcal{A}^{\dag}\mathcal{A}}$ and $p,q\in[1,\infty]$ such that $\frac{1}{p} + \frac{1}{q}=1$. Here, we have taken $p=1$, $q=\infty$, and $\norm{\mathcal{A}}_{\infty}=\norm{\mathcal{A}}_{\rm op}$ is the operator norm of $\mathcal{A}$ given by its largest singular value. Thus, we have
\begin{align}
   \frac{ C^2_{K}(t)}{4 \norm{\mathcal{S}}^2_{\rm op}}\leq1-\tr(\mathcal{P}_t\mathcal{P}_0)= 1 - \frac{\bradked{\rho_0}{\rho_t}^2}{(\tr \rho^2_0)^2}.
\end{align}

\section{Speed of evolution for open quantum dynamics in terms of eigenvalues and eigenmodes of Liouvillian}\label{Exp}
The spectral decomposition of a time-independent Liouvillian \(\mathcal{L}\) is given by
\begin{equation}
\mathcal{L} = \sum_{i=0}^{d^{2}-1} \lambda_i \kedbrad{r_i}{l_i},
\end{equation}
where $\lambda_i$ (in general complex) are eigenvalues of $\cal{L}$, and \(\{\ked{r_i}\}\) and \(\{\ked{l_i}\}\) denote the right and left eigenvectors of \(\mathcal{L}\), respectively. The matrices $\{r_i\}$ and $\{l_i\}$ corresponding to \(\{\ked{r_i}\}\) and \(\{\ked{l_i}\}\) are regarded as the right and left eigenmatrices of the Liouvillian. These eigenvectors form a bi-orthonormal basis satisfying the relation \(\bradked{l_i}{r_j} = \delta_{ij}\). It is known that a CPTP dynamics generated by a time-independent Liouvillian $\cal{L}$ has a unique steady (stationary) state if \(\lambda_0=0\) is non-degenerate. This steady state of the system is represented by the density matrix \(\rho_{ss}\), which satisfies \(\mathcal{L}|\rho_{ss}) = 0\), where, \(|\rho_{ss})=|r_{0})\) corresponds to the right eigenvector of the Liouvillian associated with the eigenvalue \(\lambda_0 = 0\)~\cite{Rivas2012}. It is known that the real parts of the eigenvalues are negative, i.e., $\Re(\lambda_{i>0}) < 0$, and can be arranged in ascending order of the absolute values of their real parts as $\abs{\Re(\lambda_0)}=0 <\abs{\Re(\lambda_1)}  < \abs{\Re(\lambda_2)} < \dots < \abs{\Re(\lambda_{d^2-1})}$~\cite{Breuer2007, Rivas2012, Can2019}. The real parts of the eigenvalues determine the relaxation rates and timescales of the system towards the steady state, and the corresponding eigenvectors $|r_{i>0})$ are called the decay modes~\cite{Albert2014, PRE2015}. The imaginary parts of the eigenvalues, $\Im(\lambda_i)$, describe the oscillatory processes that may take place during the evolution. We can then write the time evolved density matrix \(\rho_t\) from an arbitrary initial state $\rho_0$ in the right eigenbasis of \(\mathcal{L}\) as
\begin{equation}\label{time}
|\rho_t) = e^{\mathcal{L} t}|\rho_0) = |\rho_{ss}) + \sum_{i=1}^{d^2-1} e^{\lambda_i t} c_i |r_i),
\end{equation}
where $c_i=\bradked{l_i}{\rho_0}$.
Each decay mode $|r_i)$ in this expansion decays with a characteristic timescale \(1/|\Re(\lambda_i)|\). In general, the slowest-decaying mode $|r_1)$, associated with $\lambda_1$, defines the longest relaxation timescale \(1/|\Re(\lambda_1)|\) for the system to reach approximately close to the steady state starting from a given generic initial state.
Similarly, the fastest-decaying mode $|r_{d^2-1})$, associated with \( \lambda_{d^2-1} \), defines the shortest relaxation timescale for the system to reach the steady state. If $c_{i}=0$ $\forall$ $i \neq d^{2}-1$, this implies that the initial state is orthogonal to all left eigenvectors $|l_i)$ corresponding to the decay modes $|r_i)$ except the left eigenvector corresponding to the fastest-decaying mode, and hence the system exhibits the shortest relaxation timescale. For a given initial state \(\rho_0\), any mode from Eq.~\eqref{time} can be eliminated from the dynamics by applying a suitable unitary such that $c_i=(l_1|U\rho_0U^{\dagger})=0$ \cite{Carollo2021, Kochsiek2022}.
 
 The evolution speed of a system under the Lindblad dynamics is given by
\begin{eqnarray}
    \Delta\mathcal{L}^2&=&{\tr(\mathcal{L}^\dagger\mathcal{L}\mathcal{P}_t)-\tr(\mathcal{L}\mathcal{P}_t)^2}\nonumber\\
    &=&\bradked{\tilde{\rho}_t}{\mathcal{L}^\dagger\mathcal{L}|\tilde{\rho}_t}-\bradked{\tilde{\rho}_t}{\mathcal{L}|\tilde{\rho}_t}^2\nonumber\\
    &=&\frac{1}{\tr(\rho_t^2)}\left(\tr(\dot{\rho}_t^\dagger\dot{\rho}_t)-\tr(\rho_t\dot{\rho}_t)^2\right),
\end{eqnarray}
where in the third line we have used $\tilde{\rho}_t=\rho_t/\sqrt{\tr(\rho_t^2)}$ and $\ked{\dot{\rho}_t}=\mathcal{L}\ked{\rho_t}$. Now, using the above two equations, we can write the evolution speed in terms of the eigenmodes and eigenvalues of the Liouvillian as
\begin{eqnarray}\label{speedeig}
\Delta\mathcal{L}^2&=&\frac{\sum_{ij}e^{(\lambda_i^*+\lambda_j)t}\lambda_i^*\lambda_jc_i^*c_j\tr(r_i^{\dagger}r_j)-\left(\sum_{i}e^{\lambda_it}\lambda_ic_i\tr(\rho_{ss}r_i)+\sum_{ij}e^{(\lambda_i^*+\lambda_j)t}\lambda_jc_i^*c_j\tr(r_i^{\dagger}r_j)\right)^2}{\tr(\rho_{ss}^2)+\sum_{i}e^{\lambda_it}c_i\tr(\rho_{ss}r_i)+\sum_{i}e^{\lambda_i^*t}c_i^*\tr(r_i^{\dagger}\rho_{ss})+\sum_{ij}e^{(\lambda_i^*+\lambda_j)t}c_i^*c_j\tr(r_i^{\dagger}r_j)},
\end{eqnarray}
where $c_i=\tr(l_i^\dagger\rho_0)=(l_i|\rho_0)$, and $*$ denotes the complex conjugate. The above expression for the speed of evolution is dictated by the exponential terms $e^{\Re(\lambda_{i})t}$ in the long time limit. The exponential terms with larger $|\Re(\lambda_{i})|$ will decay and vanish faster compared to the terms with relatively smaller $|\Re(\lambda_{i})|$. It is important to note that the oscillatory effects due to $e^{\Im(\lambda_{i})t}$ will be negligible in the long time limit. If for a given initial state all $c_{i\neq d^{2}-1}=0$, then only the fastest decaying mode $|r_{d^2-1})$ survives, in this scenario the system approaches the steady state with the maximal possible speed. The speed of evolution in this scenario is always higher than any other initial state which activates the relatively slower decaying modes. Here, mode activation refers to the contribution of the specific eigenmodes of the Liouvillian to the evolution of a quantum state. As already mentioned above, any mode from Eq.~\eqref{time} can be eliminated from the dynamics by applying a suitable unitary on the initial state such that $c_i=(l_i|U\rho_0U^{\dagger})=0$~\cite{Carollo2021, Kochsiek2022}. Thus, in general, the speed of evolution can be enhanced by eliminating the slower decaying modes by rotating the initial state by a unitary. Let us now consider the case (as discussed in the main text) where the Liouvillian $\mathcal{L}$ induces the optimal dynamics that connects the initial state $\rho_0$ and its orthogonal state $\rho_0^{\perp}$ via geodesic. If $\rho_0^{\perp}=\rho_{ss}$ is the steady state of the system, then we know that the time-evolved state that traces the geodesic in Liouville space can be expressed as (see Eq.~\eqref{Line})
\begin{eqnarray}\label{geodesicbiortho}
|\rho_t) &=& P_t|\rho_0) + (1 - P_t)|\rho_{ss})\nonumber\\
&=&|\rho_{ss})+P_t(|\rho_0)-|\rho_{ss}))\nonumber\\
&=&|\rho_{ss})+P_t\sum_{i}c_i|r_i),
\end{eqnarray}
where in the third line, we have used the expansion of $\rho_0$ in the eigenbasis of $\mathcal{L}$. Comparing Eq.~\eqref{time} and Eq.~\eqref{geodesicbiortho}, we obtain
\begin{equation}
    \sum_{i}(P_t-e^{\lambda_it})c_i|r_i)=0.
\end{equation}
Since $|r_i)$ are linearly independent, we get $(P_t-e^{\lambda_it})c_i=0~\forall~i\geq1$. It implies that either $c_i=(l_i|\rho_0)=0$ or $P_t=e^{\lambda_it}$. If in addition to evolving along the geodesic the system also has the maximal speed of evolution, then $c_{i}=0~\forall~i\neq{d^2-1}$. Therefore, in such a case the time-evolved state can be written as 
\begin{eqnarray}
|\rho_t) &=&|\rho_{ss})+ e^{-\lambda_{m}t} c_{d^{2}-1}|r_{d^{2}-1}),
\end{eqnarray}
where we have $\lambda_{d^{2}-1}=-\lambda_{m}$, with $\lambda_{m}>0$. In the above equation, we assume that $\lambda_{d^{2}-1}$ is non-degenerate; however, for cases where it is degenerate, similar reasoning can be applied. In this case, the evolution speed of the system is
\begin{eqnarray}
\Delta\mathcal{L}^2&=&\frac{e^{-2\lambda_{m}t}\abs{\lambda_m}^2\abs{c_m}^2\tr(r_m^{\dagger}r_m)-\left(e^{-\lambda_{m}t}\lambda_mc_m\tr(\rho_{ss}r_m)+e^{-2\lambda_{m}t}\lambda_m\abs{c_m}^2\tr(r_m^{\dagger}r_m)\right)^2}{\tr(\rho_{ss}^2)+e^{-\lambda_{m}t}c_m\tr(\rho_{ss}r_m)+e^{-\lambda_{m}t}c_m^*\tr(r_m^{\dagger}\rho_{ss})+e^{-2\lambda_{m}t}\abs{c_m}^2\tr(r_m^{\dagger}r_m)},
\end{eqnarray}
where $m$ represents $d^{2}-1$.

In general, the distance between the initial and time-evolved state can be expressed as
\begin{align}\label{disteig}
\Theta(\rho_0, \rho_t)&= \arccos{\left(\frac{(\rho_0|\rho_t)}{\sqrt{\tr{(\rho_0^2)}}\sqrt{\tr{(\rho_t^2)}}} \right)} \nonumber \\
&= \arccos{\left(\frac{\tr{(\rho_{ss}^2)} + \sum_{i} c_i \tr{(\rho_{ss}r_i)}+\sum_{i}c_i^{*}\tr{(r_i^\dagger\rho_{ss})}+\sum_{i,j}c_i^{*}c_je^{\lambda_jt}\tr{(r_i^\dagger r_j)}}{\sqrt{\tr{(\rho_0^2)}}\sqrt{\tr(\rho_{ss}^2)+\sum_{i}e^{\lambda_it}c_i\tr(\rho_{ss}r_i)+\sum_{i}e^{\lambda_i^*t}c_i^*\tr(r_i^{\dagger}\rho_{ss})+\sum_{ij}e^{(\lambda_i^*+\lambda_j)t}c_i^*c_j\tr(r_i^{\dagger}r_j)}}\right)}.
\end{align}

Hence, using Eqs.~\eqref{speedeig} and \eqref{disteig}, it is established that $T_{QSL}=\Theta(\rho_0, \rho_T)/\langle\!\langle \Delta \mathcal{L} \rangle\!\rangle_{T}$ (lower bound in Eq.\eqref{LMT}) is dictated by the eigenvalues of Liouvillian and the overlap of the initial state with the left eigenvectors of Liouvillian.

\section{Mpemba effect in Thermalizing (Amplitude damping) process}\label{appendixG}
Let us consider a qubit system (described by the Hamiltonian $H=-\frac{\omega}{2} \sigma_z$) interacting with a thermal bath. The dynamics of the qubit system can be described by the following Lindblad master equation~\cite{Breuer2007,lidar2020} 
\begin{align}\label{amplitude}
    \frac{d\rho_t}{dt} =& \gamma(1+n)\left(\sigma_{-}\rho_t \sigma_{+}- \frac{1}{2}\{\sigma_{+}\sigma_{-},\rho_t\}\right)+\gamma n \left(\sigma_{+}\rho_t \sigma_{-}- \frac{1}{2}\{\sigma_{-}\sigma_{+},\rho_t\}\right),
\end{align}
where $\sigma_{-} = \ketbra{0}{1}$ and $\sigma_{+}=\ketbra{1}{0}$ are the lowering and raising operators, $\gamma$ is Weiskopf-Wigner decay constant, $n=1/{(e^{\beta \omega}-1)}$ is the average number of photons in the bath with frequency $\omega$, and $\beta=1/k_{B}\tau$ is the inverse temperature of the bath. The bath at zero temperature corresponds to $n=0$. Let the initial state of the system be $\rho_0 = \ketbra{\psi}{\psi}$, where $\ket{\psi} = \alpha \ket{0}+\sqrt{1-\alpha^2}\ket{1}$.
Now, by solving the Lindblad master equation~\eqref{amplitude}, we obtain the state of the system at time $t$
\begin{equation}
\rho_t=\left(
\begin{array}{cc}
 \frac{n+e^{-(2 n+1) t \gamma } \left(\alpha ^2+n \left(2 \alpha ^2-1\right)-1\right)+1}{2 n+1} & e^{-\frac{1}{2} (2 n+1) t \gamma } \alpha  \sqrt{1-\alpha ^2} \\
 e^{-\frac{1}{2} (2 n+1) t \gamma } \alpha  \sqrt{1-\alpha ^2} & \frac{e^{-(2 n+1) t \gamma } \left(-\alpha ^2+n \left(-2 \alpha ^2+e^{(2 n+1) t \gamma }+1\right)+1\right)}{2 n+1} \\
\end{array}
\right).
\end{equation}
The distance between the initial and time-evolved state is given by
\begin{equation}
  \Theta(\rho_{0},\rho_{t})= \arcsec\left(\frac{(2 n+1) e^{\gamma  (2 n+1) t} \sqrt{\frac{e^{-2 \gamma  (2 n+1) t} f(t,n,\gamma)}{(2 n+1)^2}+\frac{1}{2} \left(\frac{1}{(2 n+1)^2}+1\right)}}{2 \alpha ^4 (2 n+1)-\alpha ^2 (4 n+3)-2 \left(\alpha ^2-1\right) \alpha ^2 (2 n+1) e^{\frac{1}{2} \gamma  (2 n+1) t}+\left(\alpha ^2+n\right) e^{\gamma  (2 n+1) t}+n+1}\right),
\end{equation}
where $f(t,n,\gamma):=\left(2 \left(\alpha ^2+\left(2 \alpha ^2-1\right) n-1\right)^2-2 \left(\alpha ^4 (2 n+1)^2-2 \alpha ^2 (n+1) (2 n+1)+n+1\right) e^{\gamma  (2 n+1) t}\right)$. The steady state of the system, which corresponds to the eigenstate of the Liouvillian with zero eigenvalue, is given by
\begin{equation}
    \rho_{ss}=\left(
\begin{array}{cc}
 \frac{n+1}{2 n+1} & 0 \\
 0 & \frac{n}{2 n+1} \\
\end{array}
\right).
\end{equation}
Hence, the distance between the steady state and time-evolved state is given by 
\begin{align}
    \Theta(\rho_{ss},\rho_{t})= \arccos\left(\frac{\alpha ^2+\left(2 \alpha ^2-1\right) n+(2 n (n+1)+1) e^{\gamma  (2 n+1) t}-1}{\sqrt{(2 n (n+1)+1) \left(2 \left(\alpha ^2+\left(2 \alpha ^2-1\right) n-1\right)^2+g(t,n,\gamma,\alpha)\right)}}\right),
\end{align}
where $g(t,n,\gamma,\alpha):=\left(-2 \left(\alpha ^4 (2 n+1)^2-2 \alpha ^2 (n+1) (2 n+1)+n+1\right) e^{\gamma  (2 n+1) t}+(2 n (n+1)+1) e^{2 \gamma  (2 n+1) t}\right)$.
\begin{figure}[htp]
\centering
\includegraphics[width=5.4cm]{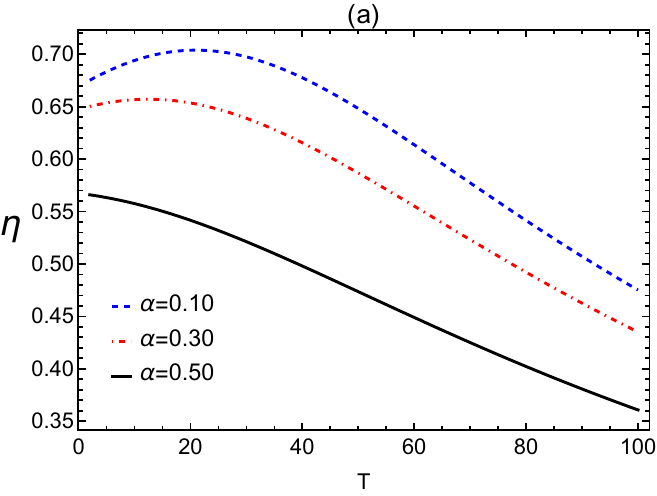}
\space
\space
\space
\includegraphics[width=5.4cm]{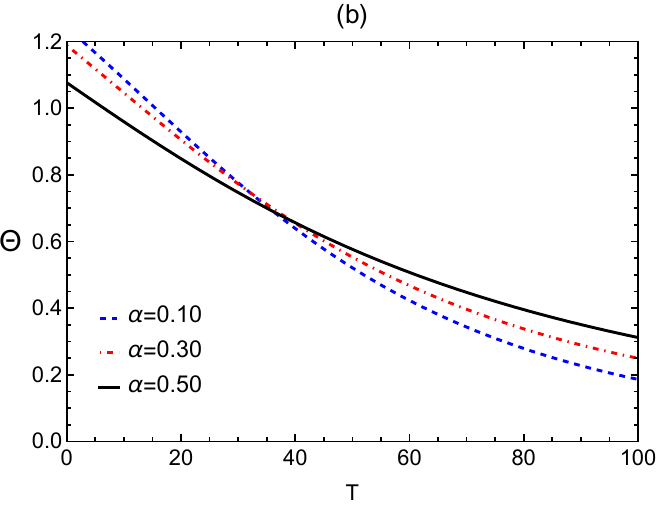}
\space
\space
\space
\includegraphics[width=5.4cm]{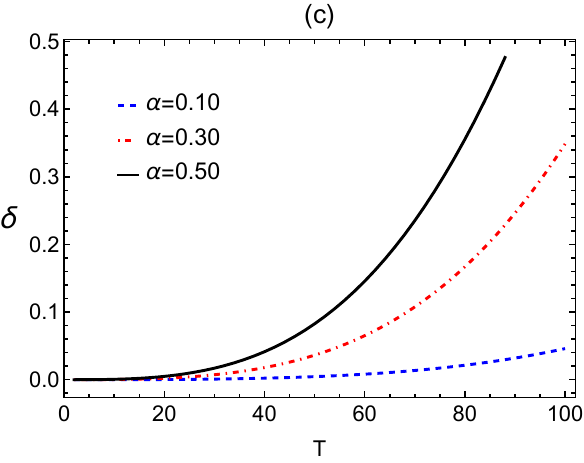}
 \caption{{\bf Mpemba effect in a qubit system under the action of the thermal environment at finite temperature ($n=0.5$).} The computations are performed for various values of $\alpha$ (which corresponds to different initial states) and the parameter $\gamma = 0.01$. (a) Speed efficiency (see Eq.~\eqref{speff} in the main text) for different initial states. (b) Distance between different non-equilibrium states and the steady state, i.e., $\Theta(\rho_{T},\rho_{ss})$. (c) Tightness of the bound given in Eq.~\ref{mainbound} (Eq.~\eqref{LMT} in the main text) for different initial states.}
\label{MpembaF}
\end{figure}

The speed of evolution of the qubit system under the action of the thermal environment is given by
 \begin{equation}
     \Delta \mathcal{L}=\frac{\gamma  (2 n+1)^2 e^{\frac{1}{2} \gamma  (2 n+1) t} \sqrt{h_1(t,n,\gamma,\alpha)-\alpha ^2 \left(\alpha ^2-1\right) (2 n (n+1)+1) e^{2 \gamma  (2 n+1) t}}}{\sqrt{2} \left(h_2(t,n,\gamma,\alpha)+(2 n (n+1)+1) e^{2 \gamma  (2 n+1) t}\right)},
 \end{equation}
where $\mathcal{L}$ is the Liouvillian corresponding to Eq.~\eqref{amplitude}, $\Delta \mathcal{L}$ is its variance in the state $\ked{\tilde{\rho}_t}$, and
\begin{eqnarray*}
    h_1(t,n,\gamma,\alpha):&=&2 \left(\alpha ^2+\left(2 \alpha ^2-1\right) n-1\right) \left(\alpha ^4+\left(2 \alpha ^2-1\right) n-1\right) e^{\gamma  (2 n+1) t}-2 \alpha ^2 \left(\alpha ^2-1\right) \left(\alpha ^2 (-(2 n+1))+n+1\right)^2,\nonumber\\
    h_2(t,n,\gamma,\alpha):&=&2 \left(\alpha ^2+\left(2 \alpha ^2-1\right) n-1\right)^2-2 \left(\alpha ^4 (2 n+1)^2-2 \alpha ^2 (n+1) (2 n+1)+n+1\right) e^{\gamma  (2 n+1) t}.
\end{eqnarray*}
The operator-norm of Liouvillian $\mathcal{L}$ is given by
\begin{eqnarray*}
\norm{\cal{L}}_{\rm op}^2=\max\{\ 2 \left(1 + 2n(1 + n)\right) \gamma^2,\ \frac{1}{4} \left(1 + 2n\right)^2\gamma^2\}.
\end{eqnarray*}

Using the above equations, we plot Fig.-\ref{MpembaF} for different initial states, which shows the Mpemba effect at $n=0.5$ (finite temperature). The explanation and the inference of the plots are exactly similar to that provided in the main text. Note that for the plots in the main text, the relevant quantities needed, which are provided in the above equations, correspond to $n=0$.

\twocolumngrid
\bibliography{name.bib}
\end{document}